\documentclass[a4paper,11pt]{article}
\usepackage[utf8]{inputenc}
\usepackage{amsmath}
\usepackage{graphicx}

\usepackage{amsmath}
\usepackage{amssymb}
\usepackage{ascmac}
\usepackage{amsthm}
\usepackage{here}
\theoremstyle{plain}
\newtheorem{theorem}{Theorem}
\newtheorem{lemma}[theorem]{Lemma}

\theoremstyle{definition}
\newtheorem{definition}[theorem]{Definition}

\usepackage{algorithm}
\usepackage{algorithmic}
\usepackage{bm}
\usepackage{url}
\usepackage{cases}
\usepackage{graphicx}
\usepackage{comment}
\usepackage{blindtext}
\usepackage{thm-restate}
\usepackage{thmtools}
\usepackage{booktabs}
\usepackage{mleftright}\mleftright
\usepackage{multicol}
\usepackage{authblk}
\usepackage{enumerate}

\usepackage{bm}
\usepackage{authblk}

\usepackage[margin=2.3cm]{geometry}
\usepackage{microtype}  
\usepackage{hyphenat}
\usepackage{natbib}

\allowdisplaybreaks

\usepackage[hidelinks]{hyperref}
\usepackage[capitalize,nameinlink]{cleveref}

\usepackage{tikz}
\usepackage{tikz-cd}
\usetikzlibrary{positioning,chains,fit,shapes,decorations.pathreplacing,calligraphy,matrix,backgrounds}

\usepackage{threeparttable}
\usepackage{color}
\usepackage{CJKutf8}
\usepackage{setspace}

\usepackage{color}
\usepackage{graphicx,colortbl,multirow}
\usepackage{subcaption}

\definecolor{myblue}{RGB}{80,80,160}
\definecolor{mygreen}{RGB}{80,160,80}

\newcommand{\roundrobincase}[1]{
    \begin{tikzpicture}
      \matrix (m) [matrix of math nodes, 
        nodes in empty cells,
        nodes={outer sep=0pt}, 
        column sep=-\pgflinewidth,
        row sep=-\pgflinewidth,
        ampersand replacement=\&,
        column 1/.style={nodes={minimum width=0.5cm, minimum height=1.0cm}},
        column 2/.style={nodes={minimum width=1.3cm, minimum height=1.0cm}},
        column 3/.style={nodes={minimum width=1cm, minimum height=1.0cm}},
        column 4/.style={nodes={minimum width=1cm, minimum height=1.0cm}}]
        {
          \mathrm{Agent}\ a_1 \& A_1^{r-1} \& g_{a_1}^{r} \& g_{a_1}^{r+1}  \\
          \mathrm{Agent}\ a_2 \& A_2^{r-1} \& g_{a_2}^{r} \& g_{a_2}^{r+1} \\
          \mathrm{Agent}\ a_3 \& A_3^{r-1} \& g_{a_3}^{r} \& g_{a_3}^{r+1} \\
        };
      \begin{pgfonlayer}{background}
        #1
      \end{pgfonlayer}
      \draw (m-1-1.north west) rectangle (m-3-4.south east);
      \draw (m-1-1.north east) -- (m-3-1.south east);
      \draw (m-1-2.north east) -- (m-3-2.south east);
      \draw (m-1-3.north east) -- (m-3-3.south east);
      \draw (m-1-1.south west) -- (m-1-4.south east);
      \draw (m-2-1.south west) -- (m-2-4.south east);
    \end{tikzpicture}
}

\usepackage[color={red!100!green!33},textsize=scriptsize]{todonotes}

\usepackage[right]{lineno}

\def\final{0}  
\def\iflong{\iffalse}
\ifnum\final=0  

\else 
\newcommand{\rnote}[1]{}

\fi

\begin{document}

\title{Position Fair Mechanisms Allocating Indivisible Goods}
\author[1]{Ryoga Mahara}
\author[2]{Ryuhei Mizutani}
\author[3]{Taihei Oki}
\author[1]{Tomohiko Yokoyama}
\affil[1]{The University of Tokyo}
\affil[ ]{\texttt{ \{\href{mailto:mahara@mist.i.u-tokyo.ac.jp}{mahara}, \href{mailto:tomohiko_yokoyama@mist.i.u-tokyo.ac.jp}{tomohiko\_yokoyama}\}@mist.i.u-tokyo.ac.jp}}
\affil[2]{Keio University}
\affil[ ]{\texttt {\href{mailto:mizutani@math.keio.ac.jp}{mizutani@math.keio.ac.jp}}}
\affil[3]{Hokkaido University, RIKEN}
\affil[ ]{\texttt {\href{mailto:oki@icredd.hokudai.ac.jp}{oki@icredd.hokudai.ac.jp}}}
%

\date{}

\maketitle


\begin{abstract}
    Fair division mechanisms for indivisible goods require agent orderings to deterministically select one allocation when running the algorithm in practice.
    We introduce position envy-freeness up to one good (PEF1) as a fairness criterion for mechanisms: a mechanism is said to satisfy PEF1 if for any pair of agent orderings, no agent prefers their bundle determined under one ordering to that under another ordering by more than the utility of a single good. First, we propose a scale-invariant, polynomial-time mechanism that satisfies PEF1 and yields an envy-freeness up to one good (EF1) allocation. For the case of two agents, we establish that any mechanism producing a maximum Nash welfare allocation eliminates envy based on positions by removing one good, provided that utilities are positive. Additionally, we present a polynomial-time mechanism based on the adjusted winner procedure, which satisfies PEF1 and produces an EF1 and Pareto optimal allocation for two agents. 
    In contrast, we demonstrate that well-known mechanisms such as round-robin and envy-cycle elimination do not generally satisfy PEF1.
\end{abstract}

\section{Introduction}\label{sec:Introduction}
Fair division of indivisible goods among agents is a fundamental problem in economics and computer science with significant practical importance.
Applications range from course allocation in universities~\citep{budish2016course} and inheritance division~\citep{GoldmanProcaccia2015} to various other real-world settings~\citep{IgarashiYokoyama2023,HanSuksompong2024}.
For surveys, see~\citep{Walsh2020FairDivision,Aziz2022Algorithmic,AmanatidisSuvvey2023}.
In fair division, agents are typically assumed to have \emph{additive} utilities over bundles of goods, and the challenge is to find an allocation that meets certain fairness criteria.
One well-studied criterion is \emph{envy-freeness}~\citep{Foley1967}, which requires an allocation to satisfy that no agent prefers another agent's bundle to their own.

In this paper, we consider fairness properties not of allocations but of \emph{mechanisms}, which have received little attention in existing work.
Particularly, we deal with fairness regarding an input order of agents. 
To illustrate this, consider a situation where two agents have identical utilities for one good, requiring mechanisms to establish clear rules for determining which agent receives it.
Namely, mechanisms must incorporate an ordering among the agents.

We formalize \emph{mechanisms} as follows (see Section~\ref{sec:Mechanism} and Figure~\ref{fig:mechanism-overview} for formal definitions). 
The input of a mechanism is a tuple of utilities arranged by an \emph{agent ordering}. 
The mechanism prioritizes agents based on their \emph{positions} in the agent ordering, and outputs an allocation. 
Under this framework, utilities that agents obtain by the mechanism can vary substantially depending on their positions in the agent ordering.

Previous work by~\citet{ManabeOkamoto2012} introduced a fairness concept for mechanisms in the \emph{divisible} goods setting, based on agent orderings. 
This concept, called \emph{meta-envy-freeness}, requires a mechanism to ensure that each agent obtains the same utility regardless of the agent ordering. A meta-envy-free mechanism always exists for divisible goods since divisible goods can be split equally among agents. Specifically, when $n$ agents desire a single good, each can simply receive $1/n$ of the good, obtaining identical utilities.

In contrast, for indivisible goods, the situation differs significantly. 
A meta-envy-free mechanism cannot exist even with two agents and a single good, as only the agent in the higher priority position can receive the good. 

A similar difficulty arises for envy-freeness in the indivisible goods setting, where envy-free allocations do not always exist.
This motivated the development of relaxations such as \emph{envy-freeness up to one good (EF1)}~\citep{Budish2011}. An allocation is said to be EF1 if any agent's envy toward another agent can be eliminated by removing a single good from the envied agent's bundle. 
Importantly, an EF1 allocation always exists for agents with additive utilities~\citep{Lipton2004,Caragiannis2019}.

Inspired by EF1 and the concept of meta-envy-freeness, we introduce \emph{position envy-freeness up to one good (PEF1)}\footnote{We use the term ``position envy-free'' instead of ``meta-envy-free'' to clarify the meaning of ``meta.''} as a fairness criterion for mechanisms with respect to agent orderings.
A mechanism is said to satisfy PEF1 if for any agent and two agent orderings, their envy
towards a bundle under one ordering over that under another ordering can be eliminated by removing a single good from the envied bundle.

Since PEF1 is a property of mechanisms concerning different agent orderings, a PEF1 mechanism may not produce a fair allocation. 
This raises a fundamental question: can we design a mechanism that satisfies PEF1 and is guaranteed to produce an EF1 allocation?

\paragraph{Our Results}

First, we answer the above question affirmatively by presenting a PEF1 mechanism that always produces an EF1 allocation (Theorem~\ref{theorem:matching_draft}).
The mechanism employs a maximum-weight matching to determine a bundle for agents in each round.
Notably, the mechanism runs in polynomial time.

Second, for two agents, we establish that any mechanism that maximizes the Nash welfare (i.e., the geometric mean of agents' utilities) satisfies PEF1.
Since a maximum Nash welfare (MNW) allocation is both EF1 and Pareto optimal (PO, i.e., no allocation makes some agent better off without making another agent worse off)~\citep{Caragiannis2019}, this mechanism outputs an allocation that satisfies both EF1 and PO.
While computing a MNW allocation is NP-hard~\citep{nguyen2014computational} and even APX-hard~\citep{Lee2017},
we present a modified version of the adjusted winner mechanism~\citep{BramsTa96,AzizAdjustedWinner2015,AzizCaIg22} that is PEF1, runs in polynomial time, and produces an EF1 and PO allocation for the case of two agents.

Finally, we analyze the round-robin mechanism, which produces an EF1 allocation~\citep{Caragiannis2019}.
We show that the mechanism satisfies PEF1 for the case of at most three agents and it lacks PEF1 when the number of agents is at least four.
In addition, we show that the envy-cycle mechanism, which returns an EF1 allocation when agents have monotone utility functions~\citep{Lipton2004}, also fails to satisfy PEF1 in general.

We remark that all mechanisms proposed in this paper are scale-invariant; namely, scaling an agent's utility by any positive constant does not affect the output.

\paragraph{Further Related Work}

A mechanism is said to be \emph{anonymous} if for any agent ordering, the outcome \emph{allocation} of the mechanism remains unchanged~\citep{Gibbard1973,Satterthwaite1975}.
In this paper, we distinguish between anonymity and position (meta-)envy-freeness.
Position envy-freeness ensures that the \emph{utilities} that agents receive do not change, while the allocation produced by the mechanism may vary depending on the agent ordering.

The \emph{equal-treatment-of-equals (ETE)} is also a fairness concept for mechanisms~\cite{Moulin2004}. A mechanism satisfies ETE if agents with identical preferences receive the same bundle of goods.

Similar to position envy-freeness, both anonymity and ETE are impossible to achieve for indivisible goods.
For divisible goods, the compatibility of anonymity and ETE with fairness and efficiency properties has been extensively studied~\citep{ShapleyScarf1974,zhou1990conjecture,BOGOMOLNAIA2001,Roth2005,BeiHuzhangSuksompong2020}.

Our work also relates to a fairness concept in allocation rules. 
An \emph{allocation rule} is a map from utility profiles to sets of allocations satisfying specified criteria~\citep{sonmez1999strategy}.
An allocation rule is \emph{essentially-single-valued} if for any utility profile and any two allocations in its output set, each agent receives equal utility from these allocations~\citep{sonmez1999strategy}.
If a mechanism always selects its output from allocations determined by an essentially-single-valued allocation rule, then the mechanism is meta-envy-free.
It is known that the MNW allocation rule, defined as a mapping to all MNW allocations, is essentially-single-valued for agents with continuous utility functions over divisible goods~\citep{DubinsSpanier1961,segal2019monotonicity}.

\section{Preliminaries}\label{sec:Preliminaries}

\subsection{Fair Division Model}\label{sec:Model}
Let $M$ be the set of $m$ goods and $N$ be the set of $n$ agents.
A subset of $M$ is termed a \emph{bundle}.
Each agent $a \in N$ has a non-negative utility function $u_a:2^M  \rightarrow \mathbb{R}_{\geq 0}$. 
For simplicity, we denote $u_a(g)$ as $u_a(\{g\})$ for each $g\in M$.
We assume that $u_a$ is \emph{additive}, that is, we have $u_a(S) =\sum_{g\in S}u_a(g)$ for any $S\subseteq M$.
A family $u = \{u_a\}_{a \in N}$ of the utility functions of all agents is called a \emph{profile}.
Let $U_{\ge 0}$ denote the set of all profiles.
An \emph{allocation} $A=\{A_a\}_{a\in N}$ is a partition of $M$ into $n$ bundles, where $A_a$ denotes the bundle of agent $a\in N$.

We now introduce a fairness concept of an allocation.
An allocation $A$ is said to be \emph{envy-free} if no agent envies any other agent, i.e., $u_a(A_a) \geq u_a(A_{a'})$ for all $a, a' \in N$.
An allocation $A$ is called \emph{envy-free up to one good} (EF1) if for all $a,a'\in N$ with $A_{a'}\neq \emptyset$, there exists a good $g\in A_{a'}$ such that $u_a(A_a) \geq u_a(A_{a'} \setminus \{g\})$.

Next, we define an efficiency concept. An allocation $A$ is said to \emph{Pareto dominate} another allocation $A'$ if $u_a(A_a)\ge u_a(A'_{a})$ for all $a \in N$ and $u_{a'}(A_{a'})> u_{a'}(A'_{a'})$ for some $a' \in N$.
An allocation $A$ is called \emph{Pareto optimal} (PO) if there is no allocation that Pareto dominates $A$.

\subsection{Mechanism and Position Fairness}\label{sec:Mechanism}

We now introduce several notions related to agent orderings.
An \emph{agent ordering} is defined as a bijection from $N$ to $[n]=\{1,2,\ldots,n\}$.
We call $\pi(a)$ the position of agent $a$ under $\pi$. 
Let $\Pi$ denote the set of all agent orderings.
An \emph{ordered profile} $u = (u_1, u_2, \dotsc, u_n)$ is an ordered $n$-tuple of utility functions.
Given an agent ordering $\pi$, let $u_\pi$ be an ordered profile generated from an original profile $u = \{u_a\}_{a \in N}$ by mapping utility functions according to $\pi$, i.e., $(u_\pi)_i = u_{\pi^{-1}(i)}$ for each $i \in [n]$.
An \emph{ordered allocation} $A = (A_1, A_2, \dotsc, A_n)$ is a partition of $M$ into $n$ bundles indexed from $1$ to $n$.

A \emph{mechanism} $\mathcal{M}$ is defined as a map from ordered profiles to ordered allocations.
More precisely, an input of a mechanism is an ordered profile $u_\pi$ generated from a profile $u$ and an agent ordering $\pi$, and the mechanism cannot access $u$ and $\pi$; in other words, the mechanism is ignorant of the correspondence between agents and their positions in $\pi$.
Then, $\mathcal{M}(u_{\pi})$ means the ordered allocation returned by $\mathcal{M}$ when the input is $u_{\pi}$.
See Figure~\ref{fig:mechanism-overview} for an illustration.

We now define a fairness concept of a mechanism concerning agent orderings, called \emph{position envy-freeness}. 
This concept means that no agent envies the bundle under one agent ordering $\pi$ compared to that under another agent ordering $\pi'$.
\begin{definition}
    A mechanism $\mathcal{M}$ satisfies \emph{position envy-freeness} with respect to $U_{\ge 0}$ if for any profile $u\in U_{\ge 0}$, for any agent orderings $\pi,\pi' \in \Pi$, and for any agent $a \in N$,
    \begin{align*}\label{eq:pef}
        u_a\big(\mathcal{M}(u_{\pi})_{\pi(a)}\big)
        \geq
        u_a\big(\mathcal{M}(u_{\pi'})_{\pi'(a)}\big).
    \end{align*}
\end{definition}
As mentioned in Section~\ref{sec:Introduction}, this concept is already known as \emph{meta-envy-freeness}~\citep{ManabeOkamoto2012} in fair division for divisible goods.

Unfortunately, a position envy-free mechanism may not exist for indivisible goods.
Consider a setting with two agents and a single good, where both agents have identical utility functions.
The input of a mechanism is an ordered tuple of the same two utility functions.
Without knowing the correspondence between agents and their positions, the mechanism must allocate the good consistently to a fixed position (e.g., position 1).
This leads to a violation of position envy-freeness since they receive the good when assigned to position 1 but not when assigned to position 2.

Following the spirit of EF1, we introduce a relaxation of the position envy-freeness, called \emph{position envy-freeness up to one good} (PEF1). 
This relaxation allows for some degree of position-based disparity, bounded by the utility of at most one good. See also Figure~\ref{fig:mechanism-overview}.
\begin{definition}
    A mechanism $\mathcal{M}$ satisfies \emph{position envy-freeness up to one good} (PEF1) with respect to $U_{\ge 0}$ if for any profile $u\in U_{\ge 0}$, agent orderings $\pi, \pi' \in \Pi$, and agent $a\in N$ with $\mathcal{M}(u_{\pi'})_{\pi'(a)} \neq \emptyset$, there exists a good $g \in \mathcal{M}(u_{\pi'})_{\pi'(a)}$ such that
    \[
        u_a\big(\mathcal{M}(u_{\pi})_{\pi(a)}\big) \geq u_a\big(\mathcal{M}(u_{\pi'})_{\pi'(a)} \setminus \{g\}\big).
    \]
\end{definition}
When we write PEF1 without further specification, we mean PEF1 with respect to $U_{\ge 0}$.

\begin{figure}[t]
    \centering
    \begin{tikzpicture}[
        node distance=2cm,
        box/.style={draw, minimum width=2.5cm, minimum height=0.8cm},
        profile/.style={draw, rounded corners, minimum width=3.5cm, minimum height=0.8cm},
        allocation/.style={draw, minimum width=3.5cm, text width=3.2cm, minimum height=1.0cm, align=center},
        arrow/.style={->, thick}
    ]
    \node[profile] (profile) at (0,3) {Profile $u = \{u_a\}_{a \in N}$};

    \node[box] (profile1) at (-2,1.6) {Ordered profile $u_{\pi}$};
    \node[box] (profile2) at (2,1.6) {Ordered profile $u_{\pi'}$};

    \draw[arrow] (profile) -- node[left=0.3cm] {$\pi$} (profile1);
    \draw[arrow] (profile) -- node[right=0.3cm] {$\pi'$} (profile2);

    \node[allocation] (alloc1) at (-2,0) {Ordered allocation $\mathcal{M}(u_{\pi})$};
    \node[allocation] (alloc2) at (2,0) {Ordered allocation $\mathcal{M}(u_{\pi'})$};

    \draw[arrow] (profile1) -- node[left] {$\mathcal{M}$} (alloc1);
    \draw[arrow] (profile2) -- node[right] {$\mathcal{M}$} (alloc2);

    \node[allocation] (final1) at (-2,-1.6) {Allocation $\{\mathcal{M}(u_{\pi})_{\pi(a)}\}_{a\in N}$};
    \node[allocation] (final2) at (2,-1.6) {Allocation $\{\mathcal{M}(u_{\pi'})_{\pi'(a)}\}_{a\in N}$};
    \draw[arrow] (alloc1) -- node[left] {} (final1);
    \draw[arrow] (alloc2) -- node[right] {} (final2);

    

    \draw[dashed] (-4.0,-2.3) rectangle (4.0,-0.9);
    \node[right] at (-3.7,-2.6) {PEF1};

    \end{tikzpicture}
    \caption{An illustration of a mechanism $\mathcal{M}$ and fairness concepts.}
    \label{fig:mechanism-overview}
\end{figure}

Additionally, we define the scale-invariance of mechanisms.
For a profile $u$ and a tuple of positive real numbers $\alpha={(\alpha_a)}_{a \in N} \in \mathbb{R}_{>0}^n$ indexed by agents in $N$, let $\alpha u$ denote the profile defined as $(\alpha u)_a(S)=\alpha_a \cdot u_{a}(S)$ for every $a \in N$ and $S\subseteq M$. 
A mechanism $\mathcal{M}$ is called \emph{scale-invariant} if $\mathcal{M}(u_\pi) = \mathcal{M}((\alpha u)_\pi)$ for every profile $u$, tuple $\alpha \in \mathbb{R}_{>0}^n$, and agent ordering $\pi \in \Pi$.

To explain these concepts, 
we present the round-robin mechanism~\citep{Caragiannis2019} described in Algorithm~\ref{alg:RR-mechanism}.
The mechanism first gives a total order of the goods for tie breaking.
Then it operates by having agents sequentially select their most preferred remaining good, following the order specified by $\pi^{-1}(1)$ to $\pi^{-1}(n)$.  
Ties are broken by choosing the good with the smallest order.
This process continues until all goods are allocated. 
The round-robin mechanism possesses two important properties: it is scale-invariant and outputs EF1 allocation for any profile and agent ordering~\citep{Caragiannis2019}.

\begin{algorithm}[h]\caption{Round-Robin Mechanism}
    \label{alg:RR-mechanism}
    \begin{algorithmic}[1]
        \REQUIRE Ordered profile $(u_1, u_2, \dotsc, u_n)$
        \ENSURE  Ordered allocation $ (A_1, A_2,\dotsc, A_n)$
        \STATE Fix indices of goods.
        \STATE $A_i \leftarrow \emptyset$ for all $i \in [n]$
        \STATE $i \gets 1$
        \WHILE{$A_1 \cup A_2 \cup \dotsb \cup A_n \subsetneq M$}
            \STATE Take $g \in \mathrm{argmax}_{g'\in M \setminus (A_1 \cup A_2 \cup \dotsb \cup A_n)}\, u_i(g')$\label{code:argmax}
            \STATE $A_i  \leftarrow A_i  \cup \{g\}$
            \STATE $i \gets (i \bmod n) + 1$
        \ENDWHILE \label{code:for-end-RR}
        \RETURN $(A_1, A_2, \dotsc, A_n)$
    \end{algorithmic}
\end{algorithm}
While the round-robin mechanism gives an EF1 allocation, it does not satisfy PEF1.
To illustrate this, consider an instance with four agents, five goods, and a profile $u$ as shown in Table~\ref{table:table_round_robin} with positive values $x>y>z>0$.
Consider two agent orderings $\pi, \pi' \in \Pi$ such that $\pi(a_i)=i$ and $\pi'(a_i) = 5-i$ for each $i\in \{1,2,3,4\}$.
Let $\mathcal{M}$ be the round-robin mechanism and let $A=\mathcal{M}(u_{\pi})$ and $B=\mathcal{M}(u_{\pi'})$.
Under $\pi$, agent $a_1$ receives $A_{\pi(a_1)}=\{g_1,g_5\}$, while under $\pi'$, agent $a_1$ receives $B_{\pi'(a_1)}=\{g_4\}$.
For any $g \in A_{\pi(a_1)}$, we have $u_{a_1}(A_{\pi(a_1)} \setminus \{g\}) > z =  u_{a_1}(B_{\pi'(a_1)})$, which violates PEF1.

\begin{table}[H]
\begin{center}
    \begin{tabular}{c|c|c|c|c|c|c}
        &$g_1$ & $g_2$ & $g_3$ & $g_4$ & $g_5$\\
          \hline
            Agent $a_1$ & $x$ & $0$ &$0$ & $z$ & $y$&\\
          \hline
            Agent $a_2$ & $0$ & $x$ & $0$ & $0$ & $y$&\\
          \hline
            Agent $a_3$ & $x$ & $0$ & $y$ &$0$  & $0$&\\
          \hline
            Agent $a_4$ & $0$ & $x$ & $z$ & $y$ & $0$&\\
    \end{tabular}
\end{center}
\caption{An example utility profile where the round-robin mechanism violates PEF1}
\label{table:table_round_robin}
\end{table}

More generally, for the round-robin mechanism $\mathcal{M}$, if $\lceil \frac{m}{n} \rceil \ge \lfloor \log_2 n \rfloor$, there exist a profile $u$ and agent orderings $\pi,\pi'$ such that even after removing any $\lfloor \log_2 n \rfloor-1$ goods from $\mathcal{M}(u_{\pi})$, the agent still prefers $\mathcal{M}(u_{\pi'})$ (Theorem~\ref{thm:degree_envy_round-robin}). 
Further discussion of the round-robin mechanism can be found in Section~\ref{sec:round-robin}.

\section{Existence of a Scale-Invariant PEF1 Mechanism Producing an EF1 Allocation}\label{sec:EF1_PEF1}

In this section, we present our main result.
\begin{theorem}\label{theorem:matching_draft}
    There exists a scale-invariant, PEF1 mechanism that always produces an EF1 allocation in polynomial time. 
\end{theorem}
We will prove Theorem \ref{theorem:matching_draft} by presenting Algorithm~\ref{alg:matching-algorithm}, which constructs a maximum-weight matching iteratively.
This mechanism is similar to one proposed by \citet{Brustle2020} in the context of fair division with subsidy. Our mechanism is distinguished by its scale invariance and the tie-breaking method.

In the mechanism, we first give an arbitrary total order of the goods $M$. Let $g_1,\ldots,g_m$ be the goods aligned according to this order.
We then ensure that $m$ is divisible by $n$ by adding dummy goods valued at zero by all agents if necessary. Initially, set $A_i=\emptyset$ for every position $i\in [n]$, and let $I$ be a set of all unallocated goods.

The core of the mechanism consists of $\frac{m}{n}$ rounds (the for loop of Lines~\ref{code:for_start}-\ref{code:for_end}).
Let $G=([n]\cup M,E)$ denote a complete bipartite graph with two disjoint vertex sets $[n]$ and $M$, where $E=\{\{i,g\}\mid i\in [n],g\in M\}$.
In each round $r$, we consider the remaining subgraph $G_r =([n]\cup I,E_r)$, where $E_r=\{\{i,g\}\mid i\in [n],g\in I\}$.

\begin{algorithm}[H]
    \caption{A Scale-Invariant and PEF1 Mechanism Producing an EF1 Allocation}
    \label{alg:matching-algorithm}
    \begin{algorithmic}[1]
        \REQUIRE Ordered profile $(u_1, u_2, \dotsc, u_n)$
        \ENSURE  Ordered allocation $ (A_1, A_2,\dotsc, A_n)$
        \STATE Fix indices of goods as $M = \{g_1,g_2,\ldots,g_m\}$.
        \STATE Add dummy goods until $m$ is divisible by $n$.
        \STATE $A_i \leftarrow \emptyset$ for all $i \in [n]$ and $I \leftarrow M$
        \FOR{$r=1$ to $\lceil \frac{m}{n} \rceil$} \label{code:for_start}
            \STATE Compute a maximum-weight matching $\mu_r$ with respect to $w$ defined by the equation~\eqref{eq:weight_function} in $G_r=([n]\cup I,E_r)$.\label{code:compute_max_matching}
            \STATE Let $\mu_r(i)$ denote the good in $I$ matched with $i\in [n]$ under $\mu_r$.
            \FOR{$i=1$ to $n$}
                \STATE $A_i \leftarrow A_i \cup \{\mu_r(i)\}$
                \STATE $I \leftarrow I \setminus \{\mu_r(i)\}$
            \ENDFOR
        \ENDFOR \label{code:for_end}
        \RETURN $(A_1,A_2, \dotsc, A_n)$
    \end{algorithmic}
\end{algorithm}

The mechanism utilizes a weight function $w: [n] \times M \to \mathbb{R}_{\ge 0}$ on $E$ defined by 
\begin{equation}\label{eq:weight_function}
    w(i,g)= 2^{m+1}n w_1(i,g) + w_2(i,g)
\end{equation}
for each edge $\{i,g\}\in E$. 
Here, two weight functions $w_1$ and $w_2$ are defined as follows.

The first weight function $w_1$ is based on utilities.
For each good $g$ and the agent in each position $i$, we define \emph{rank} $R(i,g)$ as follows: $R(i,g)=k$ if $g$ has the $k$th highest utility among all goods in $M$ for the agent in position $i$.
If multiple goods have the same utility, they are assigned the same rank, and the next rank is assigned as if no ties occurred.
For example, if four goods $g_1,g_2,g_3,g_4$ have utilities 10, 10, 8, and 7, respectively, then $R(i,g_1)=R(i,g_2)=1,R(i,g_3)=2$, and $R(i,g_4)=3$.
For each $i\in [n]$ and $g\in M$, we define
$$
w_1(i,g) = m-R(i,g).
$$
The second weight function $w_2$ is based on the total order of the goods: for each $i\in [n]$ and $\ell\in [m]$, we define 
$$
w_2(i,g_\ell) = 2^{m-\ell}.
$$
Additionally, recall the fundamental concept from matching theory. A subset of edges is a \emph{matching} if no two edges in the subset share a common vertex.
For a matching $\mu$, let $w(\mu)$ denote the total weight of the matching with respect to $w$.
Let $\mu(i)$ denote the good matched with $i$ under $\mu$.

In each $r$th round, the mechanism computes a maximum-weight matching $\mu_r$ with respect to $w$ in $G_r$ (Line~\ref{code:compute_max_matching}).
While there may be multiple maximum-weight matchings, any such matching can be selected. Then, according to $\mu$, goods are allocated to each position.
As we will show in Lemma~\ref{lemma:weight_function}, maximizing $w$ ensures that we first maximize $w_1$, then $w_2$ among matchings maximizing $w_1$. 


To illustrate the behavior of Algorithm~\ref{alg:matching-algorithm}, we apply it to the same instance from Table~\ref{table:table_round_robin}. Recall the profile with four agents and five goods with utilities satisfying $x > y > z > 0$. For this profile, the edge weights of the bipartite graph are given as in Table~\ref{table:table_matching_alg}.
\begin{table}[H]
\begin{center}
    \begin{tabular}{c|c|c|c|c|c|c}
        &$g_1$ & $g_2$ & $g_3$ & $g_4$ & $g_5$\\
          \hline
            Agent $a_1$ & 1040 & 264 &260 & 514 & 769 &\\
          \hline
            Agent $a_2$ & 528 & 1032 & 516 & 514 & 769 &\\
          \hline
            Agent $a_3$ & 1040 & 520 & 772 & 514  & 513\\
          \hline
            Agent $a_4$ & 272 & 1032 & 516 & 770 & 257 &\\
    \end{tabular}
\end{center}
\caption{Edge weights for the bipartite graph in Algorithm~\ref{alg:matching-algorithm}}
\label{table:table_matching_alg}
\end{table}
For any agent ordering $\pi$, in the bipartite graph, there exists a unique maximum-weight matching $\{(\pi(a_1),g_1), (\pi(a_2),g_2), (\pi(a_3),g_3), (\pi(a_4),g_4)\}$. This means that regardless of which agent ordering is chosen, each agent receives the same good in round 1.
In the second round, only good $g_5$ remains unallocated. 
Here, the maximum-weight matchings between the remaining good and agents are $\{(\pi(a_1), g_5)\}$ and $\{(\pi(a_2), g_5)\}$.
The algorithm selects either maximum-weight matching. For instance, we can choose to prioritize the agent with the smaller position value under $\pi$. 
Although agent $a_1$ experiences position envy between different orderings, this envy is bounded by at most one good, satisfying PEF1.

%

\subsection{Proof of Theorem~\ref{theorem:matching_draft}}\label{sec:main-proof}

The following lemma shows that a single weight function can achieve lexicographical maximization.
\begin{lemma}\label{lemma:weight_function}
    A maximum-weight matching of $G_r$ with respect to $w$ maximizes $w_1$ first, and among all such matchings, maximizes $w_2$.
\end{lemma}
\begin{proof}
Let $\mu$ and $\nu$ be any two matchings, and let $w_i(\mu)$ and $w_i(\nu)$ denote the total weights of $\mu$ and $\nu$ with respect to $w_i$ for $i\in\{1,2\}$. 
Since each matching contains at most $n$ edges, and $w_2(e) \le 2^m$ for all $e\in E$, we have $w_2(\mu) - w_2(\nu) \ge - 2^{m}n$.
If $w_1(\mu) > w_1(\nu)$, then we get $w_1(\mu) - w_1(\nu) \ge 1$ and
\begin{align*}
w(\mu) - w(\nu) 
= 2^{m+1}n (w_1(\mu) - w_1(\nu)) + (w_2(\mu) - w_2(\nu))  
\geq 2^{m+1}n - 2^{m}n  > 0.
\end{align*}
This implies that a maximum-weight matching with respect to $w$ must maximize $w_1$.
If $w_1(\mu) = w_1(\nu)$ and $w_2(\mu) > w_2(\nu)$, then $w(\mu) > w(\nu)$. Thus, among matchings maximizing $w_1$, a maximum-weight matching with respect to $w$ must maximize $w_2$.
\end{proof}

For each agent ordering $\pi\in \Pi$ and each $r\in \{1,2,\ldots,\lceil \frac{m}{n} \rceil \}$, let $\mu_{\pi,r}$ denote a maximum-weight matching in $G_r$ with respect to $w$ computed in $r$th round when the input is $u_{\pi}$.
Let $\Gamma_{\pi,r}$ be the set of goods matched by $\mu_{\pi,r}$.
For agent $a\in N$, we denote by $g_{\pi,r,a}$ the good matched with position $\pi(a)$ under $\mu_{\pi,r}$. 
By the definition of $w_2$, we can show that the set $\Gamma_{\pi,r}$ is determined independently of the agent ordering $\pi$.
\begin{lemma}\label{lemma:Gamma_independent}
    For every round $r\in \{1,2,\ldots,\lceil \frac{m}{n} \rceil \}$ and any pair of agent orderings $\pi,\pi'\in \Pi$, we have $\Gamma_{\pi,r}=\Gamma_{\pi',r}$.
\end{lemma}
\begin{proof}
    We prove the lemma by induction on $r$.
    Fix any ordering $\pi$ and, for $r=1$, suppose there exist two maximum-weight matchings $\mu_{\pi,1}$ and $\mu'_{\pi,1}$ with respect to $w$. 
    By the definition of $w$, these matchings must have the same $w_2$ weight sum.
    Since $w_2$ uses powers of $2$ based on the goods' indices, this equality implies that $\mu$ and $\mu'$ match the same set of goods. 
    Thus, $\Gamma_{\pi,1}$ is unique.
    Moreover, since $w_2$ does not depend on positions, $\Gamma_{\pi,1}$ is independent of $\pi$. 
    The same argument applies inductively for each subsequent round $r > 1$, completing the proof.
\end{proof}
Finally, we prove Theorem~\ref{theorem:matching_draft}.
\begin{proof}[Proof of Theorem~\ref{theorem:matching_draft}]
    Let $\mathcal{M}$ denote Algorithm~\ref{alg:matching-algorithm}.
    We first prove that $\mathcal{M}$ satisfies PEF1. 
    To this end, we compare any pair of two agent orderings $\pi,\pi'\in \Pi$. 
    By Lemma~\ref{lemma:Gamma_independent}, 
    $\Gamma_{\pi,r} = \Gamma_{\pi',r}$ for all $r=1,2,\ldots, \lceil \frac{m}{n} \rceil$.
    This implies that $g_{\pi, r+1, a} \notin \bigcup_{r'=1}^r \Gamma_{\pi, r'} = \bigcup_{r'=1}^r \Gamma_{\pi', r'}$ for any $r=1,2,\ldots, \lceil \frac{m}{n} \rceil-1$ and any agent $a\in N$.
    By Lemma~\ref{lemma:weight_function}, matching $\mu_{\pi', r}$ is a maximum-weight matching with respect to $w_1$.
    Thus, for any agent $a\in N$, we obtain $w_1(\pi(a), g_{\pi, r+1, a}) \le w_1(\pi'(a), g_{\pi', r, a})$ since otherwise good $g_{\pi, r+1, a}$ is included in the maximum-weight matching $\mu_{\pi', r}$.
    By the definition of $w_1$, this implies $u_a\bigl(g_{\pi, r+1, a}\bigr) \le u_a\bigl(g_{\pi', r, a}\bigr)$.
    This leads that
    \begin{align*}
        &u_a\big(\mathcal{M}(u_{\pi'})_{\pi'(a)}\big) 
        = \sum_{r=1}^{\lceil \frac{m}{n} \rceil} u_a\bigl(g_{\pi',r,a}\bigr) \ge \sum_{r=1}^{\lceil \frac{m}{n} \rceil - 1} u_a\bigl(g_{\pi',r,a}\bigr) \\
        &\ge \sum_{r=2}^{\lceil \frac{m}{n} \rceil} u_a\bigl(g_{\pi, r, a}\bigr) = u_a\big(\mathcal{M}(u_{\pi})_{\pi(a)} \setminus \{g_{\pi, 1, a}\}\big),
    \end{align*}
    which implies that the mechanism is PEF1.

    Next, we show that the mechanism always produces an EF1 allocation. 
    Fix any agent ordering $\pi$. 
    Since we choose a maximum-weight matching with respect to $w_1$, we have $u_a\bigl(g_{\pi,r+1,a'}\bigr) \leq u_a\bigl(g_{\pi, r, a}\bigr)$ for any two agents $a,a' \in N$ and $r=1,2,\ldots,\lceil \frac{m}{n} \rceil -1$. 
    Then, for all $a,a' \in N$, we have
    \begin{align*}
        &u_a\big(\mathcal{M}(u_{\pi})_{\pi(a)}\big) 
        = 
            \sum_{r=1}^{\lceil \frac{m}{n} \rceil} u_a\bigl(g_{\pi, r, a}\bigr)
        \ge  
            \sum_{r=1}^{\lceil \frac{m}{n} \rceil -1} u_a\bigl(g_{\pi, r, a}\bigr) \\
        &\ge   
            \sum_{r=2}^{\lceil \frac{m}{n} \rceil} u_a\bigl(g_{\pi, r, a'}\bigr) 
        = 
            u_a\big(\mathcal{M}(u_{\pi})_{\pi(a')}  \setminus \{g_{\pi, 1, a'}\}\big).
    \end{align*}

    We finally consider the time complexity and scale-invariance of the mechanism.
    In each round, we can find a maximum-weight matching with respect to $w$ in polynomial time~\citep{LovaszPlummer2009}.
    Thus, the mechanism runs in polynomial time.
    Furthermore, the weight function $w$ is unchanged if the profile is multiplied by a tuple of positive reals. Therefore, the mechanism is scale-invariant.
\end{proof}

\section{The Case of Two Agents}\label{sec:two_agents}
In this section, we focus on the case of two agents. 

\subsection{Maximize Nash Welfare}

The \emph{Nash welfare} of an allocation $A=\{A_a\}_{a\in N}$ is defined as $\mathrm{NW}(A) = \left(\prod_{a\in N} u_a(A_a) \right)^{1/n}$. 
An allocation $A$ is said to be \emph{maximum Nash welfare (MNW)} if it maximizes $\mathrm{NW}(A)$ among all allocations. 
Let $U_{>0}$ be the class of profiles where $u_a: 2^M \to \mathbb{R}_{>0}$ for all agents $a \in N$.

We will show that PEF1 can be achieved by a mechanism that maximizes the Nash welfare for two agents. 
To this end, we first prove the following theorem.
\begin{theorem}\label{theorem:MNW}
    When $n=2$,
    for any $u\in U_{>0}$, any two MNW allocations $A$ and $B$, and any agent $a\in N$, if $B_a\neq \emptyset$, then there exists $g\in B_a$ such that $u_a(A_a) \ge u_a(B_a \setminus \{g\})$.
\end{theorem}

\begin{proof}
    Let $N=\{a_1,a_2\}$ denote the set of two agents.
    Let $A= \{A_{a_1},A_{a_2}\}$ and $B= \{B_{a_1},B_{a_2}\}$ be two distinct MNW allocations.

    We show that for agent $a_1$, if $B_{a_1}\neq \emptyset$, there exists some good $g\in B_{a_1}$ such that $u_{a_1}(A_{a_1})\ge u_{a_1}(B_{a_1}\setminus g)$. The same argument can be applied to $a_2$. 
    Without loss of generality, we can assume that $B_{a_1} \setminus A_{a_1} \neq \emptyset$.
    Suppose, towards a contradiction, that $u_{a_1}(A_{a_1})< u_{a_1}(B_{a_1}\setminus \{g\})$ for every $g\in B_{a_1}\setminus A_{a_1}$.
    Take a good $h\in B_{a_1}\setminus A_{a_1}$ (note that $h\in A_{a_2}$).
    Consider allocations $A'$ and $B'$ where $A'_{a_1}=A_{a_1}\cup \{h\}$, $A'_{a_2}=A_{a_2}\setminus \{h\}$, $B'_{a_1}=B_{a_1}\setminus \{h\}$, and $B'_{a_2}=B_{a_2}\cup \{h\}$.
    
    We will show that $\mathrm{NW}(A') \mathrm{NW}(B')>\mathrm{NW}(A) \mathrm{NW}(B)$, contradicting the optimality of $A$ and $B$.
    Observe that
    \begin{align*}
            \dfrac{\mathrm{NW}(A')^2}{\mathrm{NW}(A)^2} \cdot \dfrac{\mathrm{NW}(B')^2}{\mathrm{NW}(B)^2}
        &=  
            \dfrac{u_{a_1}(A_{a_1}')}{u_{a_1}(A_{a_1})} \cdot\dfrac{u_{a_2}(A_{a_2}')}{u_{a_2}(A_{a_2})} \cdot\dfrac{u_{a_1}(B_{a_1}')}{u_{a_1}(B_{a_1})} \cdot\dfrac{u_{a_2}(B_{a_2}')}{u_{a_2}(B_{a_2})}\\
        &=
            \dfrac{\left(u_{a_1}(A_{a_1})+u_{a_1}(h)\right) \left(u_{a_1}(B_{a_1})-u_{a_1}(h)\right)}{u_{a_1}(A_{a_1}) u_{a_1}(B_{a_1})}\\
        &\quad
            \cdot \dfrac{\left(u_{a_2}(A_{a_2})-u_{a_2}(h)\right) \left(u_{a_2}(B_{a_2})+u_{a_2}(h)\right)}{u_{a_2}(A_{a_2}) u_{a_2}(B_{a_2})}\\
        &=
            \left(1+\dfrac{u_{a_1}(h)(u_{a_1}(B_{a_1})-u_{a_1}(A_{a_1})-u_{a_1}(h))}{u_{a_1}(A_{a_1})u_{a_1}(B_{a_1})}\right) \\
        &\quad
        \cdot \left(1+\dfrac{u_{a_2}(h)(u_{a_2}(A_{a_2})-u_{a_2}(B_{a_2})-u_{a_2}(h))}{u_{a_2}(A_{a_2})u_{a_2}(B_{a_2})}\right).
    \end{align*}
    By our assumption, $u_{a_1}(A_{a_1})<u_{a_1}(B_{a_1}\setminus \{h\})$. This implies that the first term in the product on the right-hand side is strictly greater than $1$.

    Since $A$ is MNW, $\mathrm{NW}(A)\ge \mathrm{NW}(B')$, implying $u_{a_1}(A_{a_1})u_{a_2}(A_{a_2})\ge u_{a_1}(B_{a_1}\setminus \{h\}) u_{a_2}(B_{a_2}\cup \{h\})$.
    Combined with $u_{a_1}(A_{a_1})< u_{a_1}(B_{a_1}\setminus \{h\})$, we obtain $u_{a_2}(A_{a_2})\ge u_{a_2}(B_{a_2}\cup \{h\})$. Thus, the second term is at least $1$. 
    This yields $\mathrm{NW}(A')^2 \mathrm{NW}(B')^2>\mathrm{NW}(A)^2 \mathrm{NW}(B)^2$, a contradiction.
    \end{proof}
Consider a mechanism that returns an MNW allocation, breaking ties according to agents' positions.
This mechanism is scale-invariant by the definition of Nash welfare.
Furthermore, the resulting allocation satisfies EF1 and PO~\citep{Caragiannis2019}.
By Theorem~\ref{theorem:MNW}, we obtain the following theorem.
\begin{theorem}
    When $n=2$, any mechanism that returns an MNW allocation is scale-invariant and satisfies PEF1 with respect to $U_{>0}$. Moreover, such allocations are EF1 and PO.
\end{theorem}

\subsection{Adjusted Winner Mechanism}
For two agents, we prove the existence of a scale-invariant and PEF1 mechanism that always produces EF1 and PO in polynomial time by considering the adjusted winner mechanism~\citep{BramsTa96,AzizAdjustedWinner2015,AzizCaIg22}.

\begin{theorem}\label{thm:position-envy-adjusted-winner}
    When $n=2$, there exists a scale-invariant, PEF1 mechanism that always returns an EF1 and PO allocation in polynomial time.
\end{theorem}

\begin{algorithm}[h]
    \caption{Adjusted Winner Mechanism}
    \label{alg:AD-mechanism}
    \begin{algorithmic}[1]
        \REQUIRE Ordered profile $(u_1, u_2)$
        \ENSURE  Ordered allocation $ (A_1, A_2)$
        \STATE Fix total order of the goods.
        \STATE Normalize utilities.
        \STATE Let $A_1 = \{ g\in M \mid u_{1}(g) \ge 0 \land u_{2}(g)=0\}$ and $A_2=\{ g\in M \mid u_{1}(g)=0 \land u_{2}(g)>0\}$.
        \STATE Let $M^+ = \{ g\in M \mid u_{1}(g)>0 \land u_{2}(g)>0\}$.
        \IF{$M^+\neq \emptyset$}
            \STATE Arrange the goods in $M^+$ in non-increasing order based on their utility ratios: $\frac{u_{1}(g_1)}{u_{2}(g_1)} \ge  \cdots  \ge \frac{u_{1}(g_\ell)}{u_{2}(g_\ell)}$.
            \STATE Find $P_1 = \{g_1,\ldots,g_{k-1}\}$, $P_2 = \{g_{k+1},\ldots,g_\ell\}$ and $g_k$ with $\lambda_1$ and $\lambda_2$ (where $\lambda_1 + \lambda_2 = 1$ and $\lambda_1,\lambda_2\ge 0$) such that $\frac{1}{u_1(M^+)}\left(u_1(P_1) + \lambda_1 u_1(g_k) \right) = \frac{1}{u_2(M^+)}\left(u_2(P_2) + \lambda_2 u_2(g_k) \right)$.
            \IF{$\lambda_{1} \ge \lambda_{2}$}\label{line:if_lambda}
                \STATE $A_{1} \leftarrow A_{1} \cup P_{1} \cup \{g_k\}$ and $A_{2} \leftarrow A_{2} \cup P_{2} $
            \ELSE\label{line:else_lambda}
                \STATE $A_{1} \leftarrow A_{1} \cup P_{1} $ and $A_{2} \leftarrow A_{2} \cup P_{2} \cup \{g_k\}$
            \ENDIF
        \ENDIF
        \RETURN $(A_{1},A_{2})$
    \end{algorithmic}
\end{algorithm}
See Algorithm~\ref{alg:AD-mechanism}. The mechanism first fixes a total order of goods and partitions the goods into three sets: $A_1$ containing goods valued only by position 1, $A_2$ containing goods valued only by position 2, and $M^+$ containing goods positively valued by both agents.
We then arrange goods in $M^+$ in non-increasing order of utility ratios $\frac{u_1(g)}{u_2(g)}$, breaking ties by the indices of goods. We denote the sequence by $g_1,g_2,\ldots,g_{\ell}$ where $\ell=|M^+|$.

We consider dividing fractionally these ordered goods using a boundary line: goods to the left of the boundary are allocated to agent in position 1, and goods to the right are allocated to agent in position 2. 
Formally,
as we move a boundary from left to right, there exists a unique $k\in [\ell]$ and parameters $\lambda_1,\lambda_2$ where allocating bundle $P_1=\{g_1,g_2,\ldots,g_{k-1}\}$ to agent in position 1, bundle $P_2=\{g_{k+1},g_{k+2},\ldots,g_\ell\}$ to agent in position 2, and splitting good $g_k$ in proportions $\lambda_1,\lambda_2$ gives equal utility to both agents. 
Specifically, $\frac{1}{u_1(S)}\left(u_1(P_1) + \lambda_1 u_1(g_k) \right) = \frac{1}{u_2(S)}\left(u_2(P_2) + \lambda_2 u_2(g_k) \right)$.
Based on the comparison of $\lambda_1$ and $\lambda_2$, we allocate the boundary good $g_k$ entirely to one of the agents (see Lines~\ref{line:if_lambda} and~\ref{line:else_lambda} in Algorithm~\ref{alg:AD-mechanism}).

From Section 3 of~\citep{AzizAdjustedWinner2015}, we have the following lemma.
\begin{lemma}[\citep{AzizAdjustedWinner2015}]\label{lemma:Envy-free-AD}
    For every pair of positions $i,j\in \{1,2\}$, $u_{i}(P_i) + \lambda_{i} u_{i}(g_k) \ge u_{i}(P_{j}) + \lambda_{j} u_{i}(g_k)$. 
    Moreover, no partition of $M^+$ between the two agents can make one agent better off without making the other agent worse off compared to either $(P_{1}\cup \{g_k\}, P_{2})$ or $(P_{1}, P_{2}\cup \{g_k\})$.
\end{lemma}
Using this lemma, we prove Theorem~\ref{thm:position-envy-adjusted-winner}.
\begin{proof}[Proof of Theorem~\ref{thm:position-envy-adjusted-winner}]
Let $N=\{a_1,a_2\}$ be the set of two agents.
For each $a \in N$, let $M_a = \{ g\in M \mid u_a(g) > 0 \land u_{a'}(g)=0 \}$ be the set of goods valued only by agent $a$ where $a'$ denotes the other agent, and let $M_0 = \{ g\in M \mid u_{a_1}(g) = u_{a_2}(g)=0 \}$. 
    
We first prove that the mechanism is PEF1. The boundary line and the proportions $\lambda_1,\lambda_2$ are determined solely by the utility ratios, independently of agent orderings.
Bundles $P_1$ and $P_2$ are determined by the boundary line, which is independent of positions. For any agent $a \in N$ and agent orderings $\pi$ and $\pi'$, we denote by $P_a$ the fixed set $P_{\pi(a)} = P_{\pi'(a)}$ of goods in $M^+$.
Under ordering $\pi$, agent $a$ in position $\pi(a)$ receives either $M_a \cup P_a$, $M_a \cup M_0 \cup P_a$, $M_a \cup P_a \cup \{g_k\}$, or $M_a \cup M_0 \cup P_a \cup \{g_k\}$.
Since $u_a(M_0)=0$, agent $a$'s utility equals either $u_a(M_{a} \cup P_a)$ or $u_a(M_{a} \cup P_a \cup \{g_k\})$, establishing PEF1.

We now show that the mechanism always returns an EF1 allocation. 
Fix an agent ordering $\pi$. Without loss of generality, we assume that $\pi(a_1)=1$ and $\pi(a_2)=2$. 
When $\lambda_1 \geq \lambda_2$, we have
\begin{align*}
    &u_{a_1}(\mathcal{M}(u_\pi)_{\pi(a_1)}) 
    = u_{a_1}(M_{a_1} \cup P_{a_1} ) + u_{a_1}(g_k) \\
    &\ge u_{a_1}(M_{a_2} \cup P_{a_2} ) 
    + \lambda_{2} u_{a_1}(g_k) 
    + (1 - \lambda_{1}) u_{a_1}(g_k)\\
    &\geq u_{a_1}(M_{a_2} \cup P_{a_2} ) 
    = u_{a_1}(\mathcal{M}(u_\pi)_{\pi(a_2)}),
\end{align*}
where we use Lemma~\ref{lemma:Envy-free-AD} for the first inequality, and
\begin{align*}
    &u_{a_2}(\mathcal{M}(u_\pi)_{\pi(a_2)}) 
    = u_{a_2}(M_{a_2} \cup P_{a_2} ) \\
    &\geq u_{a_2}(M_{a_1} \cup P_{a_1} ) 
    + \lambda_{1} u_{a_2}(g_k) 
    - \lambda_{2} u_{a_2}(g_k) \\
    &\geq u_{a_2}(M_{a_1} \cup P_{a_1} ) 
    = u_{a_2}(\mathcal{M}(u_\pi)_{\pi(a_1)} \setminus \{g_k\}),
\end{align*}
where we use Lemma~\ref{lemma:Envy-free-AD} for the first inequality, and $\lambda_1\ge \lambda_2$ for the second inequality.
When $\lambda_1 < \lambda_2$, the resulting allocation can be proven to be EF1 by an analogous argument to the case of $\lambda_1 \ge \lambda_2$.

Next, we prove that the mechanism produces Pareto optimal allocations. 
Let $\{A_{a_1},A_{a_2}\}$ be the allocation induced from the ordered allocation produced by the mechanism. By Lemma~\ref{lemma:Envy-free-AD}, since goods in $M_a$ are valued only by agent $a$ for each $a \in N$, and goods in $M_0$ are valued by neither agent, any allocation that differs from $\{A_{a_1},A_{a_2}\}$ would make at least one agent worse off.
Finally, the mechanism clearly runs in polynomial time, and its scale-invariance follows from the normalization step.
\end{proof}

\section{Further Analysis for Round-Robin Mechanism}\label{sec:round-robin}

We now investigate the round-robin mechanism (recall Algorithm~\ref{alg:RR-mechanism}).
We refer to each iteration of the while loop in Algorithm~\ref{alg:RR-mechanism} as a \emph{round}, in which agents select their most preferred good from the remaining goods according to a fixed agent ordering.
As mentioned in Section~\ref{sec:Preliminaries}, for the round-robin mechanism, more goods need to be removed to eliminate position-based envy as $n$ increases.
\begin{theorem}\label{thm:degree_envy_round-robin}
    When $\lceil m/n \rceil \geq \lfloor \log_2 n \rfloor$, for the round-robin mechanism, there exists a profile and two agent orderings $\pi,\pi'$ where even after removing any $\lfloor \log_2 n \rfloor-1$ goods from their bundle under $\pi$, an agent prefers keeping their remaining bundle to receiving their bundle under $\pi'$.
\end{theorem}
We prove the theorem by constructing a specific profile. The detailed proof is deferred to the Appendix.
\begin{proof}[Proof Sketch]
    Consider two agent orderings $\pi, \pi' \in \Pi$ such that $\pi(a_i)=i$ and $\pi'(a_i) = n+1-i$ for each $i\in [n]$.
    The key idea is to construct utilities where agent $a_1$'s first $\lfloor \log_2 n \rfloor$ choices under $\pi$ have significantly higher value (value $C$) than the remaining choices, while carefully setting the utilities of other agents to ensure that under $\pi'$, agent $a_1$ cannot obtain any of these highly valued goods.
    This construction ensures that even after removing any $\lfloor \log_2 n \rfloor - 1$ goods from agent $a_1$'s bundle under $\pi$, at least one good of value $C$ remains, making this bundle more valuable than their bundle under $\pi'$.
\end{proof}
Theorem~\ref{thm:degree_envy_round-robin} implies that the round-robin mechanism is not PEF1 when $n\geq 4$ and $m \geq n+1$.
When $m \le n$, the mechanism outputs an allocation where each agent receives at most one good, thus ensuring PEF1.
Moreover, we prove that the round-robin mechanism is PEF1 for two or three agents. We defer the proof to the Appendix.
\begin{theorem}\label{thm:round-robin-n-2-3}
    When $n\in \{2,3\}$, the round-robin mechanism is PEF1.  
\end{theorem}

\section{Envy-Cycle Mechanism}\label{sec:envy-cycle}

Next, we study the envy-cycle mechanism, which always produces an EF1 allocation when agents have monotone utility functions~\citep{Lipton2004}. 
Here, $u_a$ is said to be \emph{monotone} if $u_a(S) \le u_a(T)$ for any $S \subseteq T \subseteq M$. We will show the mechanism may not satisfy PEF1. All proofs are presented in the Appendix.

To describe the mechanism, we define several concepts.
We say that $P=(P_1,P_2,\ldots, P_n)$ is an \emph{ordered partial allocation} if $\bigcup_{i\in [n]} P_i \subseteq M$ and $P_{i} \cap P_{i'} = \emptyset$ for all $i\neq i' \in [n]$.
For an ordered profile $u$ and a partial allocation $P=(P_1,P_2,\ldots, P_n)$, \emph{envy graph} is defined as a directed graph $G_P=([n],E)$, where the vertex set is $[n]$ and the edge set is $E = \{(i,i') \mid i, i' \in [n], i \neq i', u_{i}(P_i) < u_{i}(P_{i'})\}$.
An \emph{envy cycle} is a directed cycle in the envy graph, that is, a sequence of positions $(i_1, i_2, \ldots, i_\ell)$ such that $(i_k, i_{k+1}) \in E$ for all $k \in [\ell]$, where $i_{\ell+1} = i_1$.

In the envy-cycle mechanism (Algorithm~\ref{alg:EC-mechanism}), we first order the goods arbitrarily. 
For each good, we first eliminate all envy cycles in the partial allocation and then allocate it to an unenvied position. 
Specifically, while there exists an envy cycle in the current allocation, we resolve it by transferring bundles in the opposite direction of the cycle. 
After eliminating all cycles, we allocate the current good to a position not envied by any other. If multiple such positions exist, we choose the one with the smallest index.
%
%

\begin{algorithm}[H]
    \caption{Envy-Cycle Mechanism}
    \label{alg:EC-mechanism}
    \begin{algorithmic}[1]
        \REQUIRE Ordered profile $(u_1, u_2, \dotsc, u_n)$
        \ENSURE  Ordered allocation $ (A_1, A_2,\dotsc, A_n)$
        \STATE Fix indices of goods as $M = \{g_1,g_2,\ldots,g_m\}$.
        \STATE Set $A_i \leftarrow \emptyset$ for all $i \in [n]$.
        \FOR{$j=1,2,\ldots,m$}
            \STATE For partial allocation $A=\{A_i\}_{i\in [n]}$, construct the envy graph $G_A$.
            \WHILE{there exists an envy cycle in $G_A$}
                \STATE Resolve the envy cycle by transferring bundles in the opposite direction of the cycle.
            \ENDWHILE
            \STATE Let $i\in [n]$ be the vertex of in-degree $0$ in $G_A$ with the smallest index.
            \STATE Set $A_i  \leftarrow A_i  \cup \{g_j\}$.
        \ENDFOR
        \RETURN $(A_1, A_2,\dotsc, A_n)$
    \end{algorithmic}
\end{algorithm}

Similar to the round-robin mechanism, the envy-cycle mechanism does not satisfy PEF1 in general. 
\begin{theorem}\label{thm:degree_envy_envy-cycle}
    When $n \ge 2$, for the envy-cycle mechanism, there exists a profile and two agent orderings $\pi,\pi'$ where even after removing any $m - \left\lfloor\frac{m}{n}\right\rfloor -1 $ goods from their bundle under $\pi$, an agent prefers keeping their remaining bundle to receiving their bundle under $\pi'$.
\end{theorem}
Theorem~\ref{thm:degree_envy_envy-cycle} implies that if the envy-cycle mechanism is PEF1, then $m - \left\lfloor\frac{m}{n}\right\rfloor \leq 1$.
Moreover, we establish that this condition is also sufficient for PEF1. 
\begin{theorem}\label{thm:envy-cycle_if_and_only_if}
    The envy-cycle mechanism is PEF1 if $m - \left\lfloor\frac{m}{n}\right\rfloor \leq 1$. 
\end{theorem}

\section{Discussion}
This paper introduces a new fairness notion, PEF1, for mechanisms. 
We demonstrate a PEF1 mechanism producing an EF1 allocation for agents with additive utilities. 
For the case of two agents, we prove the existence of a scale-invariant, PEF1 mechanism that outputs an EF1 and PO allocation.
 
Several questions remain open for future research. 
While we have shown in Theorem~\ref{theorem:MNW} that for $n=2$, the utility difference between any pair of MNW allocations is bounded by some good's utility for each agent, the result for $n > 2$ remains unknown.
We conjecture this bound holds for any $n$. Additionally, the existence of a PEF1 mechanism producing an EF1 and PO allocation for any $n$ remains an open question.

\section*{Acknowledgments}
This work was partially supported by JST ERATO Grant Number JPMJER2301.
Ryoga Mahara acknowledges additional support from JSPS KAKENHI Grant Number JP23K19956.
Ryuhei Mizutani acknowledges additional support from JSPS KAKENHI Grant Number JP23KJ0379 and JST SPRING Grant Number JPMJSP2108.
Taihei Oki acknowledges additional support from JST FOREST Grant Number JPMJFR232L.
We thank the anonymous AAAI 2026 reviewers for their valuable feedback.

\bibliographystyle{plainnat}
\bibliography{arxivsubmit_20251108}

\appendix

\section{Proof of Theorem~\ref{thm:degree_envy_round-robin}}\label{app:Proof_of_degree_envy_round-robin}

\begin{proof}[Proof of Theorem~\ref{thm:degree_envy_round-robin}]
We first observe that for any two positive integers $\ell$ and $n$, the equality $\Big\lfloor\frac{\big\lfloor\frac{n}{2^{\ell}}\big\rfloor}{2}\Big\rfloor = \big\lfloor\frac{n}{2^{\ell+1}}\big\rfloor$ holds.
    We prove the theorem by constructing an instance. Let $N=\{a_1,a_2,\ldots, a_n\}$ denote the set of agents.

    Consider two agent orderings $\pi, \pi' \in \Pi$ where $\pi(a_i) = i$ and $\pi'(a_i) = n+1-i$ for each $i\in\{1,2,\ldots,n\}$.
    For each agent $a$ and round $r$, let $g_a^r$ denote the good selected by agent $a$ in round $r$ under $\pi$.
    Let $C$ be a positive constant satisfying $C > \lceil m/n \rceil - \lfloor \log_2 n \rfloor + 1 \geq 1$.
    We construct a profile $(u_a(g))_{a\in N, g\in M}$ as follows:    
    \begin{itemize}
        \item Agent $a_1$ values their first $\lfloor \log_2 n \rfloor$ choices with utility $C$: $u_{a_1}(g_{a_1}^r) = C$ for $r=1,2,\ldots,\lfloor \log_2 n \rfloor$.
        \item Agent $a_1$ values their remaining choices with utility 1: $u_{a_1}(g_{a_1}^r) = 1$ for $r=\lfloor \log_2 n \rfloor+1,\lfloor \log_2 n \rfloor+2,\ldots,\lceil m/n \rceil$.
        \item All other agents value their choices with utility 1: $u_{a_i}(g_{a_i}^r) = 1$ for $r=1,2,\ldots,\lceil m/n \rceil$ and $i=2,3,\ldots,n$.
        \item For each round $r=1,2,\ldots,\lfloor \log_2 n \rfloor$, we set $u_{a_i}(g_{a_j}^{r}) = 1$ for indices $i$ satisfying $\Big\lfloor \frac{n}{2^{r}} \Big\rfloor+1 \leq i \leq \Big\lfloor \frac{n}{2^{r-1}} \Big\rfloor$, where $j=\big\lfloor \frac{n}{2^{r-1}} \big\rfloor-i+1$.
        \item All other utilities are set to 0.
    \end{itemize}

To aid understanding, we present an example of this profile for $n=5$ and abundant goods in Table~\ref{table:rr_n_5}.

\begin{table*}[h]
    \centering
        \begin{tabular}{c|c|c|c|c|c|c|c|c|c|c|c|c|c|c|c}
          &$g_{a_1}^1$ & $g_{a_2}^1$ & $g_{a_3}^1$ & $g_{a_4}^1$ & $g_{a_5}^1$ & $g_{a_1}^2$ & $g_{a_2}^2$ & $g_{a_3}^2$ & $g_{a_4}^2$ & $g_{5}^2$ & $g_{a_1}^3$ & $g_{a_2}^3$  & $g_{a_3}^3$ & $g_{a_4}^3$ & $g_{a_5}^3$  \\
          \hline
          Agent $a_1$ & $C$ &  & &  &  & $C$ & &  &  &  & $1$ & &  &  &  \\
          \hline
           Agent $a_2$ &  & $1$ &  &  &  & $1$ & $1$ &  &  &  &  & $1$ &  &  &  \\
          \hline
          Agent $a_3$ &  &  & $1$ &  &  &  & & $1$ &  &  &  & & $1$ &  &    \\
          \hline
          Agent $a_4$ &  & $1$ &  & $1$ &  &  & &  & $1$ &  &  & &  & $1$ &   \\
          \hline
          Agent $a_5$ & $1$ &  &  &  & $1$ &  &  &  &  & $1$ & & &  &  & $1$  \\
        \end{tabular}
        \caption{An illustration of the utility profile for $n=5$. Each blank entry represents a utility of $0$. Under ordering $\pi$, agent $a_1$ receives $\{g_{a_1}^1, g_{a_1}^2, g_{a_1}^3\}$, while under $\pi'$, they receive only $g_{a_1}^3$.}
        \label{table:rr_n_5}
\end{table*}

Under agent ordering $\pi$, agent $a_1$ obtains 
$
    \mathcal{M}(u_{\pi})_{\pi(a_1)} = \left\{g_{a_1}^1,g_{a_1}^2,\ldots,g_{a_1}^{\lceil m/n \rceil}\right\}$.
We will demonstrate that for agent $a_1$, removing any $\lfloor \log_2 n \rfloor - 1$ highest-valued goods from $\mathcal{M}(u_{\pi})_{\pi(a_1)}$ still yields a bundle with strictly greater utility than $\mathcal{M}(u_{\pi'})_{\pi'(a_1)}$.

We now consider the bundle $\mathcal{M}(u_{\pi'})_{\pi'(a_1)}$.
In the mechanism under $\pi'$, agent $a_i$ selects good $g_{a_{n-i+1}}^1$ for each $i=n,n-1,\ldots,\lfloor \frac{n}{2} \rfloor+1$ in the first round.
Next, since good $g_{a_{\lfloor \frac{n}{2} \rfloor}}^1$ have already been selected, agent $a_{\lfloor \frac{n}{2} \rfloor}$ choose good $g_{a_1}^2$.
Inductively, by the forth condition, 
for each $r=1,2,\ldots,\lfloor \log_2 n \rfloor$,
and for each $i=\Big\lfloor \frac{n}{2^{r-1}} \Big\rfloor, \Big\lfloor \frac{n}{2^{r-1}} \Big\rfloor -1,\ldots,\Big\lfloor \frac{n}{2^{r}} \Big\rfloor+1$,
agent $a_i$ selects good $g_j^r$
where $j=\big\lfloor \frac{n}{2^{r-1}} \big\rfloor-i+1$.
When $r=\lfloor \log_2 n \rfloor$,
good $g_{a_1}^{\lfloor \log_2 n \rfloor}$ 
is selected by agent $a_{j_1}$ where
$
j_1 = \left\lfloor \frac{n}{2^{\lfloor \log_2 n \rfloor-1}} \right\rfloor 
- \left\lfloor \frac{n}{2^{\lfloor \log_2 n \rfloor}} \right\rfloor +1 > 1
$ 
since we have $\left\lfloor \frac{n}{2^{\lfloor \log_2 n \rfloor}} \right\rfloor = 1$.
Therefore, agent $a_1$ chooses good $g_{a_1}^{\lfloor \log_2 n \rfloor + 1}$ in the first round. 
From the above discussion, we obtain $g_{a_1}^{\lfloor \log_2 n \rfloor} \notin \mathcal{M}(u_{\pi'})_{\pi'(a_1)}$ and
\[
    \left\{g_{a_1}^{\lfloor \log_2 n \rfloor + 1}, g_{a_1}^{\lfloor \log_2 n \rfloor + 2}, \ldots, g_{a_1}^{\lceil m/n \rceil}\right\} \supseteq \mathcal{M}(u_{\pi'})_{\pi'(a_1)}.
\]
From the first condition of the profile, we have
\begin{align*}
    u_{a_1}\left(g_{a_1}^{\lfloor \log_2 n \rfloor}\right)
    &= C \\
    &> \lceil m/n \rceil - \lfloor \log_2 n \rfloor + 1 \\
    &> \lceil m/n \rceil - \lfloor \log_2 n \rfloor \\
    &= \sum_{r=\lfloor \log_2 n \rfloor+1}^{\lceil m/n \rceil} u_{a_1}(g_{a_1}^r) \\
    &\geq u_{a_1}(\mathcal{M}(u_{\pi'})_{\pi'(a_1)}).
\end{align*}
Since 
$
g_{a_1}^{\lfloor \log_2 n \rfloor} \in \mathcal{M}(u_{\pi})_{\pi(a_1)}\setminus 
\bigl\{g_{a_1}^1,g_{a_1}^2,\ldots,g_{a_1}^{\lfloor \log_2 n \rfloor-1} \bigr\},
$
we obtain
\begin{align*}
    u_1\left(\mathcal{M}(u_{\pi})_{\pi(a_1)}\setminus \bigl\{g_{a_1}^1,g_{a_1}^2,\ldots,g_{a_1}^{\lfloor \log_2 n \rfloor-1} \bigr\} \right) 
    \ge u_{a_1}\left(g_{a_1}^{\lfloor \log_2 n \rfloor}\right),
\end{align*}
and then
\begin{align*}
    \min_{S\subseteq \mathcal{M}(u_{\pi})_{\pi(a_1)} \text{ with } |S|=\lfloor \log_2 n \rfloor-1} u_{a_1}(\mathcal{M}(u_{\pi})_{\pi(a_1)}\setminus S) 
    &\quad= u_{a_1}\left(\mathcal{M}(u_{\pi})_{\pi(a_1)}\setminus \bigl\{g_1^1,g_1^2,\ldots,g_1^{\lfloor \log_2 n \rfloor-1} \bigr\} \right) \\
    &\quad\ge u_{a_1}\left(g_{1}^{\lfloor \log_2 n \rfloor}\right) \\
    &\quad> u_{a_1}(\mathcal{M}(u_{\pi'})_{\pi'(a_1)}).
\end{align*}
This shows that even after removing any $\lfloor \log_2 n \rfloor-1$ goods from agent $a_1$'s bundle under $\pi$, agent $a_1$ still prefers their bundle under $\pi$ to their bundle under $\pi'$. This completes the proof.
\end{proof}

\section{Proof of Theorem~\ref{thm:round-robin-n-2-3} for Two Agents}\label{app:Proof_two_agents}

To prove the theorem for two agents, we first present a lemma (Lemma~\ref{lemma:round-robin-for-n=2}) that characterizes the relationship between sets of allocated goods under two different agent orderings.

Let $N=\{a_1,a_2\}$ be the set of two agents.
For proving PEF1, it suffices to consider the agent orderings $\pi$ and $\pi'$ such that $\pi(a_i) = i$ and $\pi'(a_i) = 3-i$ for each $i\in\{1,2\}$.
For each agent $a$ and round $k$, let $g_a^k$ and $h_a^k$ denote the goods selected by agent $a$ in round $k$ under agent orderings $\pi$ and $\pi'$, respectively.
Let $A^r = \bigcup_{a\in N}\{g_a^{r'} \mid r'=1,2,\ldots,r\}$ and $B^r = \bigcup_{a\in N}\{h_a^{r'} \mid r'=1,2,\ldots,r\}$.
\begin{lemma}\label{lemma:round-robin-for-n=2}
    For each $r=1,2,\ldots,\lceil m/n \rceil - 1$, 
    $B^{r}=A^{r}$ or 
    $
        B^r = A^{r} \setminus \{g_{a_2}^r\} \cup \{g_{a_1}^{r+1}\}.
    $
\end{lemma}

\begin{proof}
    We prove the statement by induction on $r$. First, consider the base case $r=1$.
    When agent $a_2$ selects the same good under $\pi'$ that agent $a_1$ selects under $\pi$ (i.e., $h_{a_2}^{1} = g_{a_1}^1$), then agent $a_1$ must select either $g_{a_2}^1$ or $g_{a_1}^2$ under $\pi'$.
    When $h_{a_2}^{1} \neq g_{a_1}^1$, we have $h_{a_2}^1=g_{a_2}^1$ and $h_{a_1}^1=g_{a_1}^1$. 
    Therefore, either $B^1 = \{h_{a_1}^1,h_{a_2}^1\} = \{g_{a_1}^1,g_{a_2}^1\} = A^1$ or $B^1 = \{h_{a_1}^1,h_{a_2}^1\} = \{g_{a_1}^1,g_{a_1}^2\} = A^{1} \setminus \{g_{a_2}^1\} \cup \{g_{a_1}^2\}$, establishing the base case.
    
    For the inductive step, suppose the statement holds for $r-1$.
    We consider two cases.
    Case (i): Suppose $B^{r-1} = A^{r-1}$. In round $r$ under $\pi'$, agent $a_2$ must select either $g_{a_1}^r$ or $g_{a_2}^r$ since these are the most preferred available goods.
    If $h_{a_2}^r=g_{a_1}^r$, then agent $a_1$ must select either $g_{a_2}^r$ or $g_{a_1}^{r+1}$. This yields either $B^{r} = A^{r}$ or $B^{r} = A^{r} \setminus \{g_{a_2}^r\} \cup \{g_{a_1}^{r+1}\}$.
    If $h_{a_2}^r = g_{a_2}^r$, then necessarily $h_{a_1}^r = g_{a_1}^r$, resulting in $B^{r} = A^{r}$.
    Case (ii): Suppose $B^{r-1} = A^{r-1} \setminus \{g_{a_2}^{r-1}\} \cup \{g_{a_1}^{r}\}$.
    In round $r$ under $\pi'$, agent $a_2$ selects $g_{a_2}^{r-1}$, which is the most preferred available good, implying $h_{a_2}^r = g_{a_2}^{r-1}$. Given this selection, agent $a_1$ selects either $g_{a_2}^r$ or $g_{a_1}^{r+1}$. When $h_{a_1}^r = g_{a_2}^r$, we obtain $B^{r} = A^{r}$. When $h_{a_1}^r = g_{a_1}^{r+1}$, we obtain $B^{r} = A^{r} \setminus \{g_{a_2}^r\} \cup \{g_{a_1}^{r+1}\}$.
\end{proof}

We now prove Theorem~\ref{thm:round-robin-n-2-3} for two agents.
\begin{proof}[Proof of Theorem~\ref{thm:round-robin-n-2-3} for two agents]
Let $\mathcal{M}$ denote the round-robin mechanism.
To prove PEF1 for two agents, it suffices to show that for every agent $a\in N$ and round $r=1,2,\ldots,\lceil m/n \rceil - 1$,
$u_a\bigl(h^{r}_a\bigr)\geq u_a\bigl(g^{r+1}_a\bigr)$ 
and 
$u_{a}\bigl(g^{r}_{a}\bigr)\geq u_{a}\bigl(h^{r+1}_{a}\bigr)$.
Indeed, if these inequalities hold, we can derive that for every $a\in N$,
$u_a(\mathcal{M}(u_{\pi'})) = \sum_{r=1}^{\lceil m/n \rceil} u_a\bigl(h^{r}_a\bigr) \ge u_a\bigl(h^{1}_a\bigr) + \sum_{r=1}^{\lceil m/n \rceil - 1} u_a\bigl(g^{r+1}_a\bigr) \ge u_a\bigl(\mathcal{M}(u_{\pi}) \setminus {g_a^1}\bigr)$, and
$u_a(\mathcal{M}(u_{\pi})) = \sum_{r=1}^{\lceil m/n \rceil} u_a\bigl(g^{r}_a\bigr) \ge u_a\bigl(g^{1}_a\bigr) + \sum_{r=1}^{\lceil m/n \rceil - 1} u_a\bigl(h^{r+1}_a\bigr) \ge u_a\bigl(\mathcal{M}(u_{\pi'}) \setminus {h_a^1}\bigr)$.

We now show these inequalities hold for every round. 
By Lemma~\ref{lemma:round-robin-for-n=2}, we consider two cases in each round $r$. 

First, consider the case where $B^{r} = A^{r}$. Here, both inequalities $u_a\bigl(h^{r}_a\bigr)\geq u_a\bigl(g^{r+1}_a\bigr)$ and $u_a\bigl(g^{r}_a\bigr)\geq u_a\bigl(h^{r+1}_a\bigr)$ hold for all $a\in N$. This follows from the fact that good $h^{r+1}_a$ is available when agent $a$ selects good $g^{r}_a$ under $\pi$, and good $g^{r+1}_a$ is available when agent $a$ selects good $h^{r}_a$ under $\pi'$.

Next, consider the case where $B^r = A^{r} \setminus \{g_{a_2}^{r}\} \cup \{g_{a_1}^{r+1}\}$. 
First, since good $g_{a_2}^{r}$ is available when agent $a_2$ makes a selection in round $r$ under $\pi'$, and agent $a_2$ chooses the good $g_{a_2}^{r}$, we have $h^{r}_{a_2} = g^{r}_{a_2}$. Since  $u_{a_2}\bigl(g^{r}_{a_2}\bigr) \ge u_{a_2}\bigl(g^{r+1}_{a_2}\bigr)$, we obtain $u_{a_2}\bigl(h^{r}_{a_2}\bigr)\ge u_{a_2}\bigl(g^{r+1}_{a_2}\bigr)$.
For agent $a_1$, by the construction, we have $h^{r}_{a_1} = g^{r+1}_{a_1}$. This implies $u_{a_1}\bigl(h^{r}_{a_1}\bigr) \ge u_{a_1}\bigl(g^{r+1}_{a_1}\bigr)$.
In round $r+1$ under $\pi'$, agent $a_2$ cannot select any good that gives higher utility than $g^{r}_{a_2}$, implying $u_{a_2}\bigl(g^{r}_{a_2}\bigr)\ge u_{a_2}\bigl(h^{r+1}_{a_2}\bigr)$. Similarly, agent $a_1$ cannot select any good that gives higher utility than $g^{r}_{a_1}$, implying $u_{a_1}\bigl(g^{r}_{a_1}\bigr)\ge u_{a_1}\bigl(h^{r+1}_{a_1}\bigr)$.
\end{proof}

\section{Proof of Theorem~\ref{thm:round-robin-n-2-3} for Three Agents}\label{app:Proof_three_agents}

Let $N=\{a_1,a_2,a_3\}$ be the set of three agents.
We fix $\pi_1$ as $\pi_1(a_i) = i$ for each $i=1,2,3$, and assume that $\pi_2 \neq \pi_1$.
For each agent $a$ and round $k$, let $g_a^k$ and $h_a^k$ denote the goods selected by agent $a$ in round $k$ under agent orderings $\pi_1$ and $\pi_2$, respectively.
Let $A^r = \bigcup_{a\in N}\{g_a^{r'} \mid r'=1,2,\ldots,r\}$ and $B^r = \bigcup_{a\in N}\{h_a^{r'} \mid r'=1,2,\ldots,r\}$.

Similar to the case of two agents, we first characterize how the sets of allocated goods differ between two agent orderings. Specifically, we show that $B^r$ can be represented as one of six cases (a), (b), (c), (d), (e) and (f) in Figure~\ref{table:round_robin_n_3_r_1}.

\begin{lemma}\label{lemma:round-robin-for-n=3}
    Suppose that $n=3$. Let $\pi_2(N) = (\pi_2(a_1),\pi_2(a_2),\pi_2(a_3))$. We have the followings:
    \begin{enumerate}
      \item When $\pi_2(N) = (1,3,2),(3,1,2)$ or $(3,2,1)$, for all $r=1,2,\ldots,\lceil m/n \rceil - 1$, we have $B^r=A^r$ (case (a)), $B^r = A^{r} \setminus \{g_{a_3}^r\} \cup \{g_{a_1}^{r+1}\}$ (case (b)), or $B^r = A^{r} \setminus \{g_{a_3}^r\} \cup \{g_{a_2}^{r+1}\}$ (case (c)).
      \item When $\pi_2(N) = (2,3,1)$, for all $r=1,2,\ldots,\lceil m/n \rceil - 1$, we have $B^r=A^r$ (case (a)), $B^r = A^{r} \setminus \{g_{a_3}^r\} \cup \{g_{a_1}^{r+1}\}$ (case (b)), or $B^r = A^{r} \setminus \{g_{a_2}^r\} \cup \{g_{a_1}^{r+1}\}$ (case (d)).
      \item When $\pi_2(N) = (2,1,3)$, for all $r=1,2,\ldots,\lceil m/n \rceil - 1$, we have $B^r=A^r$ (case (a)), 
        $B^r = A^{r} \setminus \{g_{a_3}^r\} \cup \{g_{a_1}^{r+1}\}$ (case (b)), 
        $B^r = A^{r} \setminus \{g_{a_2}^r\} \cup \{g_{a_1}^{r+1}\}$ (case (d)), 
        $B^r = A^{r} \setminus \{g_{a_2}^r\} \cup \{g_{a_2}^{r+1}\}$ (case (e)), or 
        $B^r = A^{r} \setminus \{g_{a_2}^r\} \cup \{g_{a_3}^{r+1}\}$ 
        (case (f)).
    \end{enumerate}
\end{lemma}

\begin{figure}[htbp]
        \definecolor{gray2}{gray}{0.9}
        \begin{tabular}{ccc}
        \begin{minipage}[b]{0.33\linewidth}
           \subcaption{}
           \label{case:a}
           \centering
           \roundrobincase{
               \fill[gray2] (m-1-2.north west) -- (m-1-3.north east) -- 
                             (m-1-3.south east) -- (m-1-2.south west) -- cycle;
                \fill[gray2] (m-2-2.north west) -- (m-2-3.north east) --
                             (m-2-3.south east) -- (m-2-2.south west) -- cycle;
                \fill[gray2] (m-3-2.north west) -- (m-3-3.north east) --
                             (m-3-3.south east) -- (m-3-2.south west) -- cycle;
           }
        \end{minipage}
    
        \begin{minipage}[b]{0.33\linewidth}
           \subcaption{}
           \centering
           \roundrobincase{
               \fill[gray2] (m-1-2.north west) -- (m-1-4.north east) -- 
                            (m-1-4.south east) -- (m-1-2.south west) -- cycle;
               \fill[gray2] (m-2-2.north west) -- (m-2-3.north east) --
                            (m-2-3.south east) -- (m-2-2.south west) -- cycle;
               \fill[gray2] (m-3-2.north west) -- (m-3-2.north east) --
                            (m-3-2.south east) -- (m-3-2.south west) -- cycle;
           }
        \end{minipage}

        \begin{minipage}[b]{0.33\linewidth}
           \subcaption{}
           \centering
           \roundrobincase{
               \fill[gray2] (m-1-2.north west) -- (m-1-3.north east) -- 
                            (m-1-3.south east) -- (m-1-2.south west) -- cycle;
               \fill[gray2] (m-2-2.north west) -- (m-2-4.north east) --
                            (m-2-4.south east) -- (m-2-2.south west) -- cycle;
               \fill[gray2] (m-3-2.north west) -- (m-3-2.north east) --
                            (m-3-2.south east) -- (m-3-2.south west) -- cycle;
           }
        \end{minipage}

        \\

        \begin{minipage}[b]{0.33\linewidth}
           \subcaption{}
           \centering
           \roundrobincase{
               \fill[gray2] (m-1-2.north west) -- (m-1-4.north east) -- 
                            (m-1-4.south east) -- (m-1-2.south west) -- cycle;
               \fill[gray2] (m-2-2.north west) -- (m-2-2.north east) --
                            (m-2-2.south east) -- (m-2-2.south west) -- cycle;
               \fill[gray2] (m-3-2.north west) -- (m-3-3.north east) --
                            (m-3-3.south east) -- (m-3-2.south west) -- cycle;
           }
        \end{minipage}

        \begin{minipage}[b]{0.33\linewidth}
           \subcaption{}
           \centering
           \roundrobincase{
               \fill[gray2] (m-1-2.north west) -- (m-1-3.north east) -- 
                            (m-1-3.south east) -- (m-1-2.south west) -- cycle;
               \fill[gray2] (m-2-2.north west) -- (m-2-2.north east) --
                            (m-2-2.south east) -- (m-2-2.south west) -- cycle;
               \fill[gray2] (m-3-2.north west) -- (m-3-3.north east) --
                            (m-3-3.south east) -- (m-3-2.south west) -- cycle;
               \fill[gray2] (m-2-4.north west) -- (m-2-4.north east) --
                            (m-2-4.south east) -- (m-2-4.south west) -- cycle;
           }
        \end{minipage}

        \begin{minipage}[b]{0.33\linewidth}
           \subcaption{}
           \centering
           \roundrobincase{
               \fill[gray2] (m-1-2.north west) -- (m-1-3.north east) -- 
                            (m-1-3.south east) -- (m-1-2.south west) -- cycle;
               \fill[gray2] (m-2-2.north west) -- (m-2-2.north east) --
                            (m-2-2.south east) -- (m-2-2.south west) -- cycle;
               \fill[gray2] (m-3-2.north west) -- (m-3-4.north east) --
                            (m-3-4.south east) -- (m-3-2.south west) -- cycle;
           }
        \end{minipage}
    
    \end{tabular}
    \caption{The table which illustrates six cases. The gray shaded region represents the set $B^{r}$.}
    \label{table:round_robin_n_3_r_1}
    
\end{figure}

\begin{proof}[Proof of Lemma~\ref{lemma:round-robin-for-n=3}]

Similarly to the proof of Lemma~\ref{lemma:round-robin-for-n=2}, we prove the theorem by induction on $r$. 
We check the statement when $r=1$ through a case analysis. 
\begin{enumerate}
        \item The case of $\pi_2(N)=(1,3,2)$. In the first round under $\pi_2$, agent $a_1$ selects the good $g_{a_1}^1$. Then, we have $h_{a_1}^1 = g_{a_1}^1$. Consequently, agent $a_3$ chooses $g_{a_2}^1$ or $g_{a_3}^1$. 
        \begin{enumerate}
        \item When $h_{a_3}^1 = g_{a_2}^1$, agent $a_2$ chooses $g_{a_3}^1$ or $g_{a_1}^2$ or $g_{a_2}^2$. Therefore, the possible patterns for selecting $B^1$ are the cases (a), (b) or (c) illustrated in Figure~\ref{table:round_robin_n_3_r_1} for $r=2$.
        \item When $h_{a_3}^1 = g_{a_3}^1$, agent $a_2$ selects good $g_{a_2}^1$, and $B^1=A^1$ (case (a)).
        \end{enumerate}

        \item The case of $\pi_2(N)=(3,1,2)$. In the first round under $\pi_2$, agent $a_3$ first selects $g_{a_1}^1$, $g_{a_2}^1$ or $g_{a_3}^1$. 
        \begin{enumerate}
        \item When $h_{a_3}^1 = g_{a_1}^1$, agent $a_1$ picks a good from $\{g_{a_2}^1,g_{a_3}^1,g_{a_1}^2\}$. 
        If $h_{a_1}^1 = g_{a_2}^1$, then $h_{a_2}^1 = g_{a_3}^1$ (case (a)), $g_{a_1}^2$ (case (b)) or $g_{a_2}^2$ (case (c)).
        If $h_{a_1}^1 = g_{a_3}^1$, then $h_{a_2}^1 = g_{a_2}^1$ (case (a)).
        If $h_{a_1}^1 = g_{a_1}^2$, then $h_{a_2}^1 = g_{a_2}^1$ (case (b)).
        \item When $h_{a_3}^1 = g_{a_2}^1$, agent $a_1$ obtains $g_{a_1}^1$, and agent $a_2$ picks $g_{a_3}^1$ (case (a)), $g_{a_1}^2$ (case (b)) or $g_{a_2}^2$ (case (c)).
        \item When $h_{a_3}^1 = g_{a_3}^1$, agent $a_1$ obtains $g_{a_1}^1$, and agent $a_2$ obtains $g_{a_2}^1$ (case (a)).
        \end{enumerate}

        \item The case of $\pi_2(N)=(3,2,1)$. In the first round under $\pi_2$, agent $a_3$ first selects $g_{a_1}^1$, $g_{a_2}^1$ or $g_{a_3}^1$. 
        \begin{enumerate}
        \item When $h_{a_3}^1 = g_{a_1}^1$, agent $a_2$ obtains $a_2^1$, and agent $a_1$ picks $g_{a_3}^1$ (case (a)) or $g_{a_1}^2$ (case (b)).
        \item When $h_{a_3}^1 = g_{a_2}^1$, agent $a_2$ picks a good from $\{g_{a_1}^1,g_{a_3}^1,g_{a_1}^2,a_2^2\}$. 
        If $h_{a_2}^1 = g_{a_1}^1$, then $h_{a_1}^1 = g_{a_3}^1$ (case (a)) or $g_{a_1}^2$ (case (b)).
        If $h_{a_2}^1 = g_{a_3}^1$, then $h_{a_1}^1 = g_{a_1}^1$ (case (a)).
        If $h_{a_2}^1 = g_{a_1}^2$, then $h_{a_1}^1 = g_{a_1}^1$ (case (b)).
        If $h_{a_2}^1 = a_2^2$, then $h_{a_1}^1 = g_{a_1}^1$ (case (c)).
        \item When $h_{a_3}^1 = g_{a_3}^1$, agent $a_2$ obtains $a_2^1$, and agent $a_1$ obtains $g_{a_1}^1$ (case (a)).
        \end{enumerate}

        \item The case of $\pi_2(N)=(2,3,1)$. In the first round under $\pi_2$, agent $a_2$ selects $g_{a_1}^1$ or $g_{a_2}^1$. 
        \begin{enumerate}
        \item When $h_{a_2}^1 = g_{a_1}^1$, agent $a_3$ next chooses $g_{a_2}^1$ or $g_{a_3}^1$. If $h_{a_3}^1 = g_{a_2}^1$, then agent $a_1$ picks $g_{a_3}^1$ (case (a)) or $g_{a_1}^2$ (case (b)). 
        If $h_{a_3}^1 = g_{a_3}^1$, then agent $a_1$ picks $a_2^1$ (case (a)) or $g_{a_1}^2$ (case (d)).
        \item When $h_{a_2}^1 = g_{a_2}^1$, agent $a_3$ next chooses $g_{a_1}^1$ or $g_{a_3}^1$. If $h_{a_3}^1 = g_{a_1}^1$, then agent $a_1$ selects $g_{a_3}^1$ (case (a)) or $g_{a_1}^2$ (case (b)).
        If $h_{a_3}^1 = g_{a_3}^1$, then $h_{a_1}^1=g_{a_1}^1$ (case (a)).
        \end{enumerate}
        
        \item The case of $\pi_2(N)=(2,1,3)$. In the first round under $\pi_2$, agent $a_2$ selects $g_{a_1}^1$ or $g_{a_2}^1$. 
        \begin{enumerate}
        \item When $h_{a_2}^1 = g_{a_1}^1$, agent $a_1$ next chooses $g_{a_2}^1$, $g_{a_3}^1$ or $g_{a_1}^2$.
        If agent $a_1$ picks $g_{a_2}^1$, then agent $a_3$ selects $g_{a_3}^1$ (cases (a)).
        If agent $a_1$ picks $g_{a_1}^2$, then agent $a_3$ selects $a_2^1$ (cases (b)) or $g_{a_3}^1$ (cases (d)).
        If agent $a_1$ selects $g_{a_3}^1$, then agent $a_3$ picks $g_{a_2}^1$, $g_{a_1}^2$, $a_{2}^2$ or $a_{3}^2$. Each case corresponds to each case (a), (d), (e) or (f).
        \item When $h_{a_2}^1 = g_{a_2}^1$, agent $a_1$ picks $g_{a_1}^1$, and agent $a_3$ picks $g_{a_3}^1$  (case (a)).
        \end{enumerate}

        \end{enumerate}

    From these discussion, the statement holds for $r=1$. 

    Next, suppose the statement holds for $r-1$.
    When $B^{r-1}=A^{r-1}$ (case (a)), the statement holds for $r$ by the same argument as for the case of $r=1$.
    We consider the case where $B^{r-1} \neq A^{r-1}$ and show the statement for $r$ by a case analysis. 

    \begin{enumerate}
            \item The case of $\pi_2(N)=(1,3,2)$.
            \begin{enumerate}
                \item When $B^{r-1} = A^{r-1} \setminus \{g_{a_3}^{r-1}\} \cup \{g_{a_1}^{r}\}$ (case (b)), agent $a_1$ first selects $g_{a_3}^{r-1}$, $g_{a_2}^{r}$, $g_{a_3}^{r}$ or $g_{a_1}^{r+1}$. 
                If $h_{a_1}^{r}=g_{a_3}^{r-1}$, then $(h_{a_3}^{r}, h_{a_2}^{r})=(g_{a_2}^{r},g_{a_3}^{r})$, $(g_{a_2}^{r}, g_{a_1}^{r+1})$, $(g_{a_2}^{r}, g_{a_2}^{r+1})$, or $(g_{a_3}^{r}, g_{a_2}^{r})$. 
                If $h_{a_1}^{r}=g_{a_2}^{r}$, then $(h_{a_3}^{r}, h_{a_2}^{r})=(g_{a_3}^{r-1},g_{a_3}^{r})$, $(g_{a_3}^{r-1},g_{a_1}^{r+1})$, $(g_{a_3}^{r-1},g_{a_2}^{r+1})$, $(g_{a_3}^{r},a_3^{r-1})$, $(g_{a_3}^{r},g_{a_1}^{r+1})$ or $(g_{a_3}^{r},g_{a_2}^{r+1})$. 
                If $h_{a_1}^{r}=g_{a_3}^{r}$, then $(h_{a_3}^{r}, h_{a_2}^{r})=(g_{a_3}^{r-1},g_{a_2}^{r})$.
                If $h_{a_1}^{r}=g_{a_1}^{r+1}$, then $(h_{a_3}^{r}, h_{a_2}^{r})=(g_{a_3}^{r-1},g_{a_2}^{r})$. Only cases (a), (b), (c) or (a) occur.
                \item When $B^{r-1} = A^{r-1} \setminus \{g_{a_3}^{r-1}\} \cup \{g_{a_2}^{r}\}$ (case (c)), agent $a_1$ first selects $g_{a_3}^{r-1}$ or $g_{a_1}^{r}$.  If $h_{a_1}^{r}=g_{a_3}^{r-1}$, then $(h_{a_3}^{r}, h_{a_2}^{r})=(g_{a_1}^{r},g_{a_3}^{r})$, $(g_{a_1}^{r}, g_{a_1}^{r+1})$, $(g_{a_1}^{r}, g_{a_2}^{r+1})$, $(g_{a_3}^{r}, g_{a_1}^{r})$, $(g_{a_3}^{r}, g_{a_1}^{r+1})$, or $(g_{a_3}^{r}, g_{a_2}^{r+1})$. Thus, only cases (a), (b) or (c) happen.
                If $h_{a_1}^{r}=g_{a_1}^{r}$, $(h_{a_3}^{r}, h_{a_2}^{r})=(g_{a_3}^{r-1},g_{a_3}^{r})$, $(g_{a_3}^{r-1},g_{a_1}^{r+1})$, or $(g_{a_3}^{r-1},g_{a_2}^{r+1})$. For both case, only cases (a), (b) or (c) happen.
            \end{enumerate}
            \item The case of $\pi_2(N)=(3,1,2)$.
            \begin{enumerate}
                \item When $B^{r-1} = A^{r-1} \setminus \{g_{a_3}^{r-1}\} \cup \{g_{a_1}^{r}\}$ (case (b)), agent $a_3$ first selects $g_{a_3}^{r-1}$. Then, $(h_{a_1}^{r}, h_{a_2}^{r})=(g_{a_2}^{r},g_{a_3}^{r})$, $(g_{a_2}^{r},g_{a_1}^{r+1})$, $(g_{a_2}^{r},g_{a_2}^{r+1})$, $(g_{a_3}^{r},g_{a_2}^{r})$, or $(g_{a_1}^{r+1},g_{a_2}^{r})$.
                \item When $B^{r-1} = A^{r-1} \setminus \{g_{a_3}^{r-1}\} \cup \{g_{a_2}^{r}\}$ (case (c)), agent $a_3$ first selects $g_{a_3}^{r-1}$. 
                Then, $(h_{a_1}^{r}, h_{a_2}^{r})=(g_{a_1}^{r},g_{a_3}^{r})$, $(g_{a_1}^{r},g_{a_1}^{r+1})$, or $(g_{a_1}^{r},g_{a_2}^{r+1})$. Only the cases (a), (b), or (c) happen.
            \end{enumerate}

            \item The case of $\pi_2(N)=(3,2,1)$.
            \begin{enumerate}
            \item When $B^{r-1} = A^{r-1} \setminus \{g_{a_3}^{r-1}\} \cup \{g_{a_1}^{r}\}$ (case (b)), agent $a_3$ first selects $g_{a_3}^{r-1}$.
            Then, $(h_{a_2}^{r}, h_{a_1}^{r})=(g_{a_2}^{r},g_{a_3}^{r})$ or $(g_{a_2}^{r},g_{a_1}^{r+1})$.
            \item When $B^{r-1} = A^{r-1} \setminus \{g_{a_3}^{r-1}\} \cup \{g_{a_2}^{r}\}$ (case (c)), agent $a_3$ first selects $g_{a_3}^{r-1}$. 
            Then, $(h_{a_2}^{r}, h_{a_1}^{r})=(g_{a_1}^{r},g_{a_3}^{r})$, $(g_{a_1}^{r},g_{a_1}^{r+1})$, $(g_{a_3}^{r},g_{a_1}^{r})$, $(g_{a_1}^{r+1},g_{a_1}^{r})$, or $(g_{a_2}^{r+1},g_{a_1}^{r})$.
            \end{enumerate}

            \item The case of $\pi_2(N)=(2,3,1)$.
            \begin{enumerate}
            \item When $B^{r-1} = A^{r-1} \setminus \{g_{a_3}^{r-1}\} \cup \{g_{a_1}^{r}\}$ (case (b)), agent $a_2$ first selects $g_{a_3}^{r-1}$ or $g_{a_2}^{r}$.
            If $h_{a_2}^{r} = g_{a_3}^{r-1}$, then $(h_{a_3}^{r}, h_{a_1}^{r}) = (g_{a_2}^{r}, g_{a_3}^{r})$ (case (a)), $(g_{a_2}^{r}, g_{a_1}^{r+1})$ (case (b)), $(g_{a_3}^{r}, g_{a_2}^{r})$ (case (a)) or $(g_{a_3}^{r}, g_{a_1}^{r+1})$ (case (d)).
            If $h_{a_2}^{r} = g_{a_2}^{r}$, then $(h_{a_3}^{r}, h_{a_1}^{r}) = (g_{a_3}^{r-1}, g_{a_3}^{r})$ or $(g_{a_3}^{r-1}, g_{a_1}^{r+1})$.
            \item When $B^{r-1} = A^{r-1} \setminus \{g_{a_2}^{r-1}\} \cup \{g_{a_1}^{r}\}$ (case (d)), agent $a_2$ first selects $g_{a_2}^{r-1}$. 
            Then, $(h_{a_3}^{r}, h_{a_1}^{r})=(g_{a_2}^{r},g_{a_3}^{r})$, $(g_{a_2}^{r},g_{a_1}^{r+1})$, or $(g_{a_3}^{r},g_{a_2}^{r})$.
            \end{enumerate}

            \item The case of $\pi_2(N)=(2,1,3)$.
            \begin{enumerate}
            \item When $B^{r-1} = A^{r-1} \setminus \{g_{a_3}^{r-1}\} \cup \{g_{a_1}^{r}\}$ (case (b)), agent $a_2$ first selects $g_{a_3}^{r-1}$ or $g_{a_2}^{r}$.
            If $h_{a_2}^{r} = g_{a_3}^{r-1}$, then $(h_{a_3}^{r}, h_{a_1}^{r}) = (g_{a_2}^{r}, g_{a_3}^{r})$, $(g_{a_2}^{r}, g_{a_1}^{r+1})$, $(g_{a_3}^{r}, g_{a_2}^{r})$, or $(g_{a_3}^{r}, g_{a_1}^{r+1})$.
            If $h_{a_2}^{r} = g_{a_2}^{r}$, then $(h_{a_3}^{r}, h_{a_1}^{r}) = (g_{a_3}^{r-1}, g_{a_3}^{r})$ or $(g_{a_3}^{r-1}, g_{a_1}^{r+1})$.
            \item When $B^{r-1} = A^{r-1} \setminus \{g_{a_2}^{r-1}\} \cup \{g_{a_1}^{r}\}$ (case (d)), agent $a_2$ first selects $g_{a_2}^{r-1}$. 
            Then, $(h_{a_3}^{r}, h_{a_1}^{r})=(g_{a_2}^{r},g_{a_3}^{r})$, $(g_{a_2}^{r},g_{a_1}^{r+1})$, $(g_{a_3}^{r},g_{a_2}^{r})$, or $(g_{a_3}^{r},g_{a_1}^{r+1})$.
            \item When $B^{r-1} = A^{r-1} \setminus \{g_{a_2}^{r-1}\} \cup \{g_{a_2}^{r}\}$ (case (e)), agent $a_2$ first selects $g_{a_2}^{r-1}$.
            Then, $(h_{a_3}^{r}, h_{a_1}^{r})=(g_{a_1}^{r},g_{a_3}^{r})$, $(g_{a_1}^{r},g_{a_1}^{r+1})$ or $(g_{a_3}^{r},g_{a_1}^{r})$.
            \item When $B^{r-1} = A^{r-1} \setminus {g_{a_2}^{r-1}} \cup {g_{a_3}^{r}}$ (case (f)), agent $a_2$ first selects $g_{a_2}^{r-1}$. After that, the pair $(h_{a_3}^{r}, h_{a_1}^{r})$ must be one of $(g_{a_1}^{r},g_{a_2}^{r})$, $(g_{a_1}^{r},g_{a_1}^{r+1})$, $(g_{a_2}^{r}, g_{a_1}^{r})$, $(g_{a_1}^{r+1}, g_{a_1}^{r})$, $(g_{a_2}^{r+1}, g_{a_1}^{r})$, or $(g_{a_3}^{r+1}, g_{a_1}^{r})$. Subsequently, $B^r$ follows one of cases (a), (d), (e), or (f).
            \end{enumerate}

        \end{enumerate}
Therefore, the statement holds for $r$ and we complete the proof.
\end{proof}

Finally, we show the proof of Theorem~\ref{thm:round-robin-n-2-3} for three agents.
\begin{proof}[Proof of Theorem~\ref{thm:round-robin-n-2-3} for three agents]
Similarly to the case of two agents, we prove that for every agent $a\in N$ and round $r=1,2,\ldots,\lceil m/n \rceil - 1$,
\[
u_a\bigl(h^{r}_a\bigr)\geq u_a\bigl(g^{r+1}_a\bigr)
\quad\text{and}\quad
u_{a}\bigl(g^{r}_{a}\bigr)\geq u_{a}\bigl(h^{r+1}_{a}\bigr).
\]

For a round $r$ such that $B^{r} = A^{r}$ (cases (a) in Figure~\ref{table:round_robin_n_3_r_1}) holds, the both inequalities hold for every $a\in N$.
Thus, we only consider a round $r$ such that there exist $i_1 \in \{2,3\} $ and $i_2 \in \{1,2,3\}$ such that $B^{r} = A^{r} \setminus \{g_{a_{i_1}}^{r}\} \cup \{g_{a_{i_2}}^{r+1}\}$ (cases (b), (c), (d), (e) and (f) in Figure~\ref{table:round_robin_n_3_r_1}). 
For each $i\in \{1,2,3\} \setminus \{i_2\}$,
we have $u_{a_i}\bigl(h^{r}_{a_i}\bigr)\ge u_{a_i}\bigl(g^{r+1}_{a_i}\bigr)$ since good $g_{a_i}^{r+1}$ remains in round $r$ under $\pi_2$, i.e., $g_{a_i}^{r+1} \notin B^r$.

To complete the proof, we consider a case analysis.
\begin{enumerate}
            \item The case of $\pi_2(N)=(1,3,2)$. 
            \begin{enumerate}
                \item When $B^{r} = A^{r} \setminus \{g_{a_3}^{r}\} \cup \{g_{a_1}^{r+1}\}$ (case (b)), from above discussion, we have $u_{a_2}\bigl(h^{r}_{a_2}\bigr)\ge u_{a_2}\bigl(g^{r+1}_{a_2}\bigr)$ and $u_{a_3}\bigl(h^{r}_{a_3}\bigr)\ge u_{a_3}\bigl(g^{r+1}_{a_3}\bigr)$. Since $\pi_2(a_1) < \pi_2(a_2)$, in round $r$ under $\pi_2$, agent $a_1$ must pick good $g_{a_1}^{r+1}$ or a good that is preferable to good $g_{a_1}^{r+1}$.
                Thus, we have $h^{r}_{a_2} = g^{r+1}_{a_2}$, $u_{a_2}\bigl(h^{r}_{a_2}\bigr)\ge u_{a_2}\bigl(g^{r+1}_{a_2}\bigr)$ and the first inequality. 
                
                For each agent $a \in N$, there is no good within $M \setminus B^r$ that is preferable to good $g_{a}^{r}$. Thus, we obtain the second inequality. 
                
                \item When $B^{r} = A^{r} \setminus \{g_{a_3}^{r}\} \cup \{g_{a_2}^{r+1}\}$ (case (c)), we have $u_{a_1}\bigl(h^{r}_{a_1}\bigr)\ge u_{a_1}\bigl(g^{r+1}_{a_1}\bigr)$ and $u_{a_3}\bigl(h^{r}_{a_3}\bigr)\ge u_{a_3}\bigl(g^{r+1}_{a_3}\bigr)$. In round $r$ under $\pi_2$, good $g_{a_2}^{r+1}$ must be selected agent $a_2$. Otherwise, agents $a_1$ or $a_3$ choose good $g_{a_2}^{r+1}$, and this contradicts that goods $g_{a_1}^{r+1}$ and $g_{a_3}^r$ remain. Thus, we have $h^{r}_{a_2} = g^{r+1}_{a_2}$, $u_{a_2}\bigl(h^{r}_{a_2}\bigr)\ge u_{a_2}\bigl(g^{r+1}_{a_2}\bigr)$ and the first inequality. 

                For each agent $a \in N$, there is no good within $M \setminus B^r$ that is preferable to good $g_{a}^{r}$. Thus, we obtain the second inequality. 

            \end{enumerate}
            \item The case of $\pi_2(N)=(3,1,2)$.
            \begin{enumerate}
                \item When $B^{r} = A^{r} \setminus \{g_{a_3}^{r}\} \cup \{g_{a_1}^{r+1}\}$ (case (b)), we obtain the two inequalities by the same discussion as that in the case $\pi_2(N)=(1,3,2)$ and case (b).

                \item When $B^{r} = A^{r} \setminus \{g_{a_3}^{r}\} \cup \{g_{a_2}^{r+1}\}$ (case (c)), we obtain the two inequalities by the same discussion as that in the case $\pi_2(N)=(1,3,2)$ and case (c).
                
            \end{enumerate}

            \item The case of $\pi_2(N)=(3,2,1)$.
            \begin{enumerate}
            \item When $B^{r} = A^{r} \setminus \{g_{a_3}^{r}\} \cup \{g_{a_1}^{r+1}\}$ (case (b)), we have $u_{a_2}\bigl(h^{r}_{a_2}\bigr)\ge u_{a_2}\bigl(g^{r+1}_{a_2}\bigr)$ and $u_{a_3}\bigl(h^{r}_{a_3}\bigr)\ge u_{a_3}\bigl(g^{r+1}_{a_3}\bigr)$.  In round $r$ under $\pi_2$, good $g_{a_1}^{r+1}$ must be selected by agent $a_1$ or $a_2$, since good $g_{a_3}^{r}$ remains and agent $a_3$ never choose good $g_{a_1}^{r+1}$.
            
            \item When $B^{r} = A^{r} \setminus \{g_{a_3}^{r}\} \cup \{g_{a_2}^{r+1}\}$ (case (c)), we obtain the two inequalities by the same discussion as that in the case $\pi_2(N)=(1,3,2)$ and case (c).
            \end{enumerate}

            \item The case of $\pi_2(N)=(2,3,1)$.
            \begin{enumerate}
            \item When $B^{r} = A^{r} \setminus \{g_{a_3}^{r}\} \cup \{g_{a_1}^{r+1}\}$ (case (b)), we have $u_{a_2}\bigl(h^{r}_{a_2}\bigr)\ge u_{a_2}\bigl(g^{r+1}_{a_2}\bigr)$ and $u_{a_3}\bigl(h^{r}_{a_3}\bigr)\ge u_{a_3}\bigl(g^{r+1}_{a_3}\bigr)$. In round $r$ under $\pi_2$, good $g_{a_1}^{r+1}$ must be selected by agents $a_1$ or $a_2$. When good $g_{a_1}^{r+1}$ is selected by agent $a_1$, we have $h^{r}_{a_1} = g_{a_1}^{r+1}$. When good $g_{a_1}^{r+1}$ is selected by agent $a_2$, good $g_{a_2}^{r}$ must have been chosen in round $r-1$ under $\pi_2$. However, when $\pi_2(N)=(2,3,1)$, we have $g_{a_2}^{r} \notin B^{r-1}$. Thus, good $g_{a_1}^{r+1}$ is selected agent $a_1$. Hence, $u_{a_1}\bigl(h^{r}_{a_1}\bigr)\ge u_{a_1}\bigl(g^{r+1}_{a_1}\bigr)$.

            For each agent $a \in N$, there is no good within $M \setminus B^r$ that is preferable to good $g_{a}^{r}$. Thus, we obtain the second inequality. 
            
            \item When $B^{r} = A^{r} \setminus \{g_{a_2}^{r}\} \cup \{g_{a_1}^{r+1}\}$ (case (d)), we have $u_{a_2}\bigl(h^{r}_{a_2}\bigr)\ge u_{a_2}\bigl(g^{r+1}_{a_2}\bigr)$ and $u_{a_3}\bigl(h^{r}_{a_3}\bigr)\ge u_{a_3}\bigl(g^{r+1}_{a_3}\bigr)$. In round $r$ under $\pi_2$, good $g_{a_1}^{r+1}$ must be selected by agents $a_1$ or $a_3$ since good $g_{a_2}^{r}$ remains and agent $a_2$ never chooses good $g_{a_1}^{r+1}$.
            When good $g_{a_1}^{r+1}$ is selected by agent $a_3$, we must have $g_{a_3}^{r} \in B^{r-1}$ or agent $a_2$ picks good $g_{a_3}^{r}$ in round $r$ under $\pi_2$. Now we have $g_{a_3}^{r} \notin B^{r-1}$ when $\pi_2(N)=(2,3,1)$, and agent $a_2$ chooses good $g_{a_3}^{r}$ in round $r$. In this case, agent $a_1$ selects good $g_{a_1}^r$ in round $r$ under $\pi_2$, and we get $u_{a_1}\bigl(h^{r}_{a_1}\bigr)\ge u_{a_1}\bigl(g^{r+1}_{a_1}\bigr)$.

            For agents $a_1$ and $a_2$, there is no good within $M \setminus B^r$ that is preferable to good $g_{a_1}^{r}$ and good $g_{a_2}^{r}$ each. Then, we get $u_{a_1}\bigl(g^{r}_{a_1}\bigr)\ge u_{a_1}\bigl(h^{r+1}_{a_1}\bigr)$ and $u_{a_1}\bigl(g^{r}_{a_1}\bigr)\ge u_{a_1}\bigl(h^{r+1}_{a_1}\bigr)$. 
            For agent $a_3$, good $g^{r}_{a_2} \in M\setminus B^r$ may be preferable to good $g_{a_3}^{r}$. However, in round $r+1$ under $\pi_2$, agent $a_2$ picks good $g^{r}_{a_2}$. Then, agent $a_3$ can not choose good $g^{r}_{a_2}$ in round $r+1$ under $\pi_2$ and we have $u_{a_3}\bigl(g^{r}_{a_3}\bigr)\ge u_{a_3}\bigl(h^{r+1}_{a_3}\bigr)$. Therefore, we obtain the second inequality.
            \end{enumerate}

            \item The case of $\pi_2(N)=(2,1,3)$.
            \begin{enumerate}
            \item When $B^{r} = A^{r} \setminus \{g_{a_3}^{r}\} \cup \{g_{a_1}^{r+1}\}$ (case (b)), we have $u_{a_2}\bigl(h^{r}_{a_2}\bigr)\ge u_{a_2}\bigl(g^{r+1}_{a_2}\bigr)$ and $u_{a_3}\bigl(h^{r}_{a_3}\bigr)\ge u_{a_3}\bigl(g^{r+1}_{a_3}\bigr)$. In round $r$ under $\pi_2$, good $g_{a_1}^{r+1}$ must be selected by agents $a_1$ or $a_2$. When good $g_{a_1}^{r+1}$ is selected by agent $a_2$, good $g_{a_2}^{r}$ must have been chosen in round $r-1$ under $\pi_2$. Thus, in round $r-1$, case (e) holds. However, in the case (e), $g_{a_2}^{r-1}\notin B^{r-1}$ and agent $a_2$ selects good $g_{a_2}^{r-1}$ in round $r$. Hence, good $g_{a_1}^{r+1}$ is not selected by agent $a_2$, and good $g_{a_1}^{r+1}$ must be selected by agents $a_1$. Therefore, we obtain $u_{a_1}\bigl(h^{r}_{a_1}\bigr)\ge u_{a_1}\bigl(g^{r+1}_{a_1}\bigr)$, and the first inequality.

            For each agent $a \in N$, there is no good within $M \setminus B^r$ that is preferable to good $g_{a}^{r}$. Thus, we obtain the second inequality.

            \item When $B^{r} = A^{r} \setminus \{g_{a_2}^{r}\} \cup \{g_{a_1}^{r+1}\}$ (case (d)), we have $u_{a_2}\bigl(h^{r}_{a_2}\bigr)\ge u_{a_2}\bigl(g^{r+1}_{a_2}\bigr)$ and $u_{a_3}\bigl(h^{r}_{a_3}\bigr)\ge u_{a_3}\bigl(g^{r+1}_{a_3}\bigr)$. In round $r$ under $\pi_2$, good $g_{a_1}^{r+1}$ must be selected by agents $a_1$ or $a_3$.
            Since $\pi_2(a_1) < \pi_2(a_3)$, in round $r$ under $\pi_2$, agent $a_1$ must pick good $g_{a_1}^{r+1}$ or a good that is preferable to good $g_{a_1}^{r+1}$. Thus, we get the first inequality.

            Moreover, we obtain the second inequalities by the same discussion as that in the case $\pi_2(N)=(2,3,1)$ and case (d).

            \item When $B^{r} = A^{r} \setminus \{g_{a_2}^{r}\} \cup \{g_{a_2}^{r+1}\}$ (case (e)), we have $u_{a_1}\bigl(h^{r}_{a_1}\bigr)\ge u_{a_1}\bigl(g^{r+1}_{a_1}\bigr)$ and $u_{a_3}\bigl(h^{r}_{a_3}\bigr)\ge u_{a_3}\bigl(g^{r+1}_{a_3}\bigr)$. Now, we have $g_{a_2}^{r} \notin B^{r}$. Thus, agent $a_2$ gets a good that is preferable to good $g_{a_1}^{r+1}$ in round $r$ under $\pi_2$. Then, we have $u_{a_2}\bigl(h^{r}_{a_2}\bigr)\ge u_{a_2}\bigl(g^{r+1}_{a_2}\bigr)$.

            For agents $a_1$ and $a_2$, there is no good within $M \setminus B^r$ that is preferable to good $g_{a_1}^{r}$ and good $g_{a_2}^{r}$ each.
            For agent $a_3$, good $g^{r}_{a_2} \in M\setminus B^r$ may be preferable to good $g_{a_3}^{r}$. However, in round $r+1$ under $\pi_2$, agent $a_2$ picks good $g^{r}_{a_2}$. Then, we have $u_{a_3}\bigl(g^{r}_{a_3}\bigr)\ge u_{a_3}\bigl(h^{r+1}_{a_3}\bigr)$. 

            \item When $B^{r} = A^{r} \setminus \{g_{a_2}^{r}\} \cup \{g_{a_3}^{r+1}\}$ (case (f)), we have $u_{a_1}\bigl(h^{r}_{a_1}\bigr)\ge u_{a_1}\bigl(g^{r+1}_{a_1}\bigr)$ and $u_{a_2}\bigl(h^{r}_{a_2}\bigr)\ge u_{a_2}\bigl(g^{r+1}_{a_2}\bigr)$. In round $r$ under $\pi_2$, good $g_{a_3}^{r+1}$ must be selected by agent $a_3$ since $g_{a_1}^{r+1}, g_{a_2}^{r+1} \notin B^r$. Thus, we have $h^{r}_{a_3} = g^{r+1}_{a_3}$ and $u_{a_3}\bigl(h^{r}_{a_3}\bigr)\ge u_{a_3}\bigl(g^{r+1}_{a_3}\bigr)$.

            Moreover, we obtain the second inequalities by the same discussion as that in the case $\pi_2(N)=(2,1,3)$ and case (e).
            \end{enumerate}
\end{enumerate}
From these discussion, we complete the proof.
\end{proof}

\section{Proof of Theorem~\ref{thm:degree_envy_envy-cycle}}\label{app:proof_EC_degree_envy}
\begin{proof}[Proof of Theorem~\ref{thm:degree_envy_envy-cycle}]
    Let $N=\{a_1, a_2, \ldots, a_n\}$ be the set of agents.
    We consider a preference profile where agent $a_1$ values all goods equally at 1, while all other $n-1$ agents have zero utility for every good.
    Consider the agent ordering $\pi_1$ where $\pi_1(a_i) = i$ for each $i\in \{1, 2, \ldots, n\}$.  Since no agent envies any other agent in this profile, the envy-cycle mechanism terminates without any exchanges, and consequently, agent $a_1$ receives all $m$ goods.
    Now consider the reverse ordering $\pi_2$ where $\pi_2(a_i) = n + 1 - i$ for each $i \in \{1, 2, \ldots, n\}$. In this case, since agent $a_1$ appears last in the ordering, agent $a_1$ ultimately receives at most $\left\lfloor\frac{m}{n}\right\rfloor$ goods after the envy-cycle mechanism completes. This concludes the proof.
\end{proof}

\section{Proof of Theorem~\ref{thm:envy-cycle_if_and_only_if}}\label{app:proof_EC_degree_envy_if_and_only_if}

\begin{proof}[Proof of Theorem~\ref{thm:envy-cycle_if_and_only_if}]
Since we only consider the case $n\ge 2$, we have $m - \left\lfloor\frac{m}{n}\right\rfloor = 1$. Then, either $m = 1$ (for any $n \geq 2$), or $m = 2$ and $n = 2$. When $m=1$, the envy-cycle mechanism clearly satisfies PEF1 since all agents obtain at most one good. 

We now consider the case where $n=2$ and $m=2$. Let $N=\{a_1, a_2\}$ be the set of agents.
    We compare two agent orderings $\pi_1$ and $\pi_2$ where $\pi_1(a_i) = i$ and $\pi_2(a_i) = 3-i$ for each $i\in \{1,2\}$. 
    Under $\pi_1$, good $g_1$ is first allocated to agent $a_1$. 
    Then, if $u_{a_2}(g_1)> 0$, then good $g_2$ is allocated to position 2 (agent $a_2$). In this case, each agent holds exactly one good, thus there is no position-based envy in the PEF1 sense.
    If $u_{a_2}(g_1) = 0$, then position 1 (agent $a_1$) obtains both $g_1$ and $g_2$. 
    In this case, if under $\pi_2$, agent $a_1$ receives no goods, then PEF1 can be violated.
    Under $\pi_2$, good $g_1$ is fist allocated to agent $a_2$. Then, only when $u_{a_1}(g_1) = 0$ does position 1 (agent $a_2$) obtain both $g_1$ and $g_2$, while agent $a_1$ receives no goods under $\pi_2$.
    However, since we have $u_{a_1}(g_1) = 0$, even though agent $a_1$ obtains two goods under $\pi_1$, the position-based envy compared to $\pi_2$ is at most the value of one good. Symmetrically, for agent $a_2$, the position-based envy under $\pi_1$ toward $\pi_2$ is also at most the value of one good. Therefore, these complete the proof.
\end{proof}

\end{document}